\documentclass[10pt,a4paper,notitlepage,aps,pra,showpacs,superscriptaddress,floatfix,nofootinbib,twocolumn,accepted=2023-12-04]{quantumarticle}
\pdfoutput=1

\usepackage{amsmath,amsfonts,amsthm,amssymb}
\usepackage{graphicx}
\usepackage{xspace}
\usepackage{float}
\usepackage{color}

\usepackage[dvipsnames]{xcolor}
\definecolor{quantumviolet}{HTML}{53257F}
\usepackage[colorlinks=true,citecolor=Emerald,linkcolor=quantumviolet,urlcolor=quantumviolet,hyperindex]{hyperref}

\usepackage[capitalize]{cleveref}
\usepackage[normalem]{ulem}
\usepackage{bbm}
\usepackage{bm}

\usepackage{appendix}
\usepackage{multirow}
\usepackage[T1]{fontenc}

\newtheorem{proposition}{Proposition}

\newtheorem{optproblem}{Optimization Problem}

\theoremstyle{definition}
\newtheorem*{definition}{Definition}

\crefname{section}{Sec.}{Secs.}
\crefname{appendix}{App.}{Apps.}
\crefname{equation}{Eq.}{Eqs.}
\crefname{figure}{Fig.}{Figs.}
\crefname{proposition}{Prop.}{Props.}
\crefname{obs}{Obs.}{Obs.}
\crefname{optproblem}{Prob.}{Probs.}

\Crefname{section}{Section}{Sections}
\Crefname{appendix}{Appendix}{Appendices}
\Crefname{equation}{Equation}{Equations}
\Crefname{figure}{Figure}{Figures}
\Crefname{proposition}{Proposition}{Propositions}
\Crefname{obs}{Observation}{Observations}
\Crefname{optproblem}{Problem}{Problems}

\newcommand{\one}{\mathbbm{1}} 
\newcommand{\id}{\text{id}} 
\newcommand{\Complexes}{\mathbbm{C}} 
\newcommand{\Naturals}{\mathbbm{N}} 

\newcommand{\cA}{\mathcal{A}}
\newcommand{\cB}{\mathcal{B}}
\newcommand{\cC}{\mathcal{C}}
\newcommand{\cE}{\mathcal{E}}
\newcommand{\cH}{\mathcal{H}}
\newcommand{\cI}{\mathcal{I}}
\newcommand{\cK}{\mathcal{K}}
\newcommand{\cL}{\mathcal{L}}
\newcommand{\cM}{\mathcal{M}}
\newcommand{\cO}{\mathcal{O}}
\newcommand{\cS}{\mathcal{S}}
\newcommand{\cT}{\mathcal{T}}
\newcommand{\cU}{\mathcal{U}}
\newcommand{\fS}{\mathfrak{S}}
\newcommand{\fT}{\mathfrak{T}}

\newcommand{\cHS}{\mathcal{H}_\text{S}}
\newcommand{\cHE}{\mathcal{H}_\text{E}}
\newcommand{\cHES}{\mathcal{H}_\text{ES}}
\newcommand{\dS}{d_\text{S}}
\newcommand{\dE}{d_\text{E}}
\newcommand{\dES}{d_\text{ES}}
\newcommand{\rhoEz}{\rho^{0}_\text{E}}
\newcommand{\rhoSz}{\rho^{0}_\text{S}}
\newcommand{\rhoE}{\rho_\text{E}}
\newcommand{\rhoS}{\rho_\text{S}}
\newcommand{\rhoES}{\rho_\text{ES}}
\newcommand{\T}{\text{T}}
\newcommand{\AIO}[2]{\text{A}_{#1}^{#2}}
\newcommand{\va}{\vec{a}}
\newcommand{\vt}{\vec{t}}
\newcommand{\vs}{\vec{s}}
\newcommand{\vu}{\vec{u}}
\renewcommand{\vr}{\vec{r}}
\newcommand{\SA}[1]{\text{A}_{#1}}
\newcommand{\SI}{\text{I}}
\newcommand{\SO}{\text{O}}
\newcommand{\SEP}{\text{SEP}}
\newcommand{\PpL}{\text{P}^{+}_{L}}
\newcommand{\PpN}{\text{P}^{+}_{N}}
\newcommand{\PhiL}{\Phi_{\AIO{1}{L}}}
\newcommand{\PhiN}{\Phi_{\AIO{1}{N}}}
\newcommand{\PhitL}{\widetilde{\Phi}_{\AIO{1}{L}}}
\newcommand{\Symn}[1]{\operatorname{Sym}_{n}(#1)}
\newcommand{\SymSep}{\operatorname{SymSEP}}
\newcommand{\Ta}{{\T_{\alpha}}}
\newenvironment{smallpmatrix}{\bigl(\begin{smallmatrix}}{\end{smallmatrix}\bigr)}

\let\union\cup
\let\intersection\cap
\newcommand{\binomial}[2]{\begin{smallpmatrix}#1 \\ #2\end{smallpmatrix}}
\newcommand{\multiset}[2]{\left( \! \left( \! \begin{smallmatrix}#1 \\ #2\end{smallmatrix} \! \right) \! \right)}
\newcommand{\Tr}[1]{\operatorname{tr}\left[ #1 \right]}
\newcommand{\TrP}[2]{\operatorname{tr}_{#1}\hspace{-0.25em}\left[ #2 \right]}
\newcommand{\floor}[1]{\left \lfloor #1 \right \rfloor}
\newcommand{\ket}[1]{\left \vert #1 \right \rangle}
\newcommand{\bra}[1]{\left \langle #1 \right \vert}
\newcommand{\braket}[2]{\left \langle #1 \middle \vert #2 \right \rangle}
\newcommand{\ketbra}[2]{\left \vert #1 \right \rangle \hspace{-0.4em} \left \langle #2 \right \vert}
\newcommand{\purestate}[1]{\left \vert #1 \right \rangle \hspace{-0.4em} \left \langle #1 \right \vert}
\renewcommand{\vec}[1]{\boldsymbol{#1}}

\newcommand{\rank}{\operatorname{rank}}
\newcommand{\range}[1]{\operatorname{ran}(#1)}
\newcommand{\kernel}[1]{\operatorname{Ker}(#1)}
\newcommand{\conv}{\operatorname{conv}}
\newcommand{\set}[1]{\left \{ \; #1 \; \right \}}
\newcommand{\cond}{\;\; \middle \vert \;\;}
\newcommand{\Choi}[1]{\mathcal{C}( #1 )}
\newcommand{\DC}{\operatorname{DC}}

\newcommand{\optimization}[5]{
	\bigskip
	\begin{optproblem}\label{#1} #2
		\begin{equation}
			\begin{split}
				\textbf{Given:}\quad & #3 \\
				\textbf{Find:}\quad & #4 \\
				\textbf{Subject to:}\quad & #5
			\end{split}
		\end{equation}
	\end{optproblem}
	\medskip
}

\begin{document}
	
\title{Witnessing environment dimension through temporal correlations}

\author{Lucas B. Vieira}
\email{lucas.vieira@oeaw.ac.at}
\affiliation{Institute for Quantum Optics and Quantum Information (IQOQI), Austrian Academy of Sciences,\\ Boltzmanngasse 3, 1090 Vienna, Austria}
\affiliation{Faculty of Physics, University of Vienna, Boltzmanngasse 5, 1090 Vienna, Austria}

\author{Simon Milz}
\email{milzs@tcd.ie}
\affiliation{School of Physics, Trinity College Dublin, Dublin 2, Ireland}
\affiliation{Faculty of Physics, University of Vienna, Boltzmanngasse 5, 1090 Vienna, Austria}
\affiliation{Institute for Quantum Optics and Quantum Information (IQOQI), Austrian Academy of Sciences,\\ Boltzmanngasse 3, 1090 Vienna, Austria}

\author{Giuseppe Vitagliano}
\email{giuseppe.vitagliano@tuwien.ac.at}
\affiliation{Vienna Center for Quantum Science and Technology, Atominstitut, TU Wien,  1020 Vienna, Austria}

\author{Costantino Budroni}
\email{costantino.budroni@unipi.it}
\affiliation{Department of Physics ``E. Fermi'' University of Pisa, Largo B. Pontecorvo 3, 56127 Pisa, Italy}
\affiliation{Faculty of Physics, University of Vienna, Boltzmanngasse 5, 1090 Vienna, Austria}
\affiliation{Institute for Quantum Optics and Quantum Information (IQOQI), Austrian Academy of Sciences,\\ Boltzmanngasse 3, 1090 Vienna, Austria}

\begin{abstract}
	We introduce a framework to compute upper bounds for temporal correlations achievable in open quantum system dynamics, obtained by repeated measurements on the system. As these correlations arise by virtue of the environment acting as a memory resource, such bounds are witnesses for the minimal dimension of an effective environment compatible with the observed statistics. These witnesses are derived from a hierarchy of semidefinite programs with guaranteed asymptotic convergence. We compute non-trivial bounds for various sequences involving a qubit system and a qubit environment, and compare the results to the best known quantum strategies producing the same outcome sequences. Our results provide a numerically tractable method to determine bounds on multi-time probability distributions in open quantum system dynamics and allow for the witnessing of effective environment dimensions through probing of the system alone.
\end{abstract}

\maketitle

\section{Introduction}

Physical systems are never truly isolated, and inevitably interact with their surrounding environment~\cite{BreuerPetruccionebook,RivasHuelgabook}. As a consequence, information within the system leaks away into its surroundings, leading to entanglement between the system and the environment. In many instances, this leaked information may be partially recovered at a later time, leading to non-Markovian dynamics, i.e., non-negligible memory effects and complex correlations in time~\cite{RivasRev2014,BreuerRev16,Milz_2021}.

Like their spatial counterpart, temporal correlations in quantum mechanics fundamentally differ from those that can be observed in the classical case. Such differences between classical and quantum temporal correlations have been noted since the works of Leggett and Garg~\cite{leggett_quantum_1985, leggett_realism_2008, emary_leggettgarg_2013, vitagliano2023}. In the presence of memory, a clear distinction between an underlying process -- carrying temporal correlations -- and the measurement process -- probing said correlations -- can be obtained by employing correlation kernels~\cite{Lindblad1979, accardi_quantum_1982} or higher order quantum maps~\cite{PollockPRA2018} for their description.

Broadly speaking, the dimension of a physical system, i.e., the number of perfectly distinguishable states, imposes fundamental constraints over the temporal correlations it is able to produce~\cite{budroni2019memory}. In this sense, the physical dimension acts as a memory resource, restricting the amount of information stored about the past that is capable of affecting the future. This memory constraint leads to different behaviors (i.e., different achievable temporal correlations) in classical, quantum, or more general physical theories~\cite{fritzNJP2010,hoffmann2018,budroni2019memory,spee2020simulating,vieira2022temporal,mao2022strucdimbound}. In particular quantum memories are known to allow for a larger set of correlations than classical ones of the same size~\cite{budroni2019memory,budroni2021ticking,vieira2022temporal}. However, this advantage only holds for restricted memory size, i.e., if the memory size is \emph{unrestricted}, all correlations compatible with a time-ordered causal structure may be achieved with either classical or quantum memories~\cite{fritzNJP2010,hoffmann2018}. Therefore, understanding the nature of dimensional constraints on observable correlations sheds light on fundamental differences between our descriptions of physical systems, as well as their connection with causality.

These limitations imposed by the dimensionality of physical systems have been exploited for the construction of ``dimension witnesses'', inequalities which, when violated, certify the minimum dimension of the system compatible with made observations~\cite{GallegoPRL2010, BrunnerPRL2013, GuehnePRA2014, BudroniPRL2014, SchildPRA2015, spee2020simulating, Sohbi2021}. In a similar spirit, our approach allows for the computation of upper bounds on the temporal correlations achievable in an open system dynamics with an environment of bounded dimension, so that any violations of these bounds certify the minimum dimension of the effective environment. Employing techniques from entanglement detection~\cite{dps2004completefamily,navascues2009}, these bounds are obtained by exploiting the inherent symmetries of the problem. In particular, we relax the \emph{a priori} non-linear problem of computing maximum joint probabilities in sequential measurements to a numerically tractable hierarchy of semidefinite programs (SDPs) \cite{CVXBook}. Both the computational accessibility of the final formulation of the problem, as well as the non-triviality of the resulting bounds, are then demonstrated for paradigmatic examples, showing that, indeed, joint probability distributions obtained from probing an open system alone provide viable means to deduce dimensional properties of the a priori experimentally inaccessible environment it is coupled to.

The paper is organized as follows. In \cref{sec:preliminary}, we introduce the temporal sequence protocol, the notion of temporal correlations, and their mathematical description. \cref{sec:opensys} discusses the application of temporal correlations to the task of characterizing open systems and their environment. \cref{sec:SDPmain} describes the semidefinite program we constructed to bound temporal correlations, with \cref{sec:numerics} presenting the numerical results we obtained in their optimization. In \cref{sec:implementation} we comment on the various challenges involved in realizing these numerical optimizations, and how we have overcome them. In \cref{sec:discussion} we comment on several additional technical details involving our work, with \cref{sec:conclusions} covering our conclusions and future outlook.

\section{The measurement protocol}\label{sec:preliminary}

\begin{figure}[t]\centering
	\includegraphics[width=1.0\linewidth]{"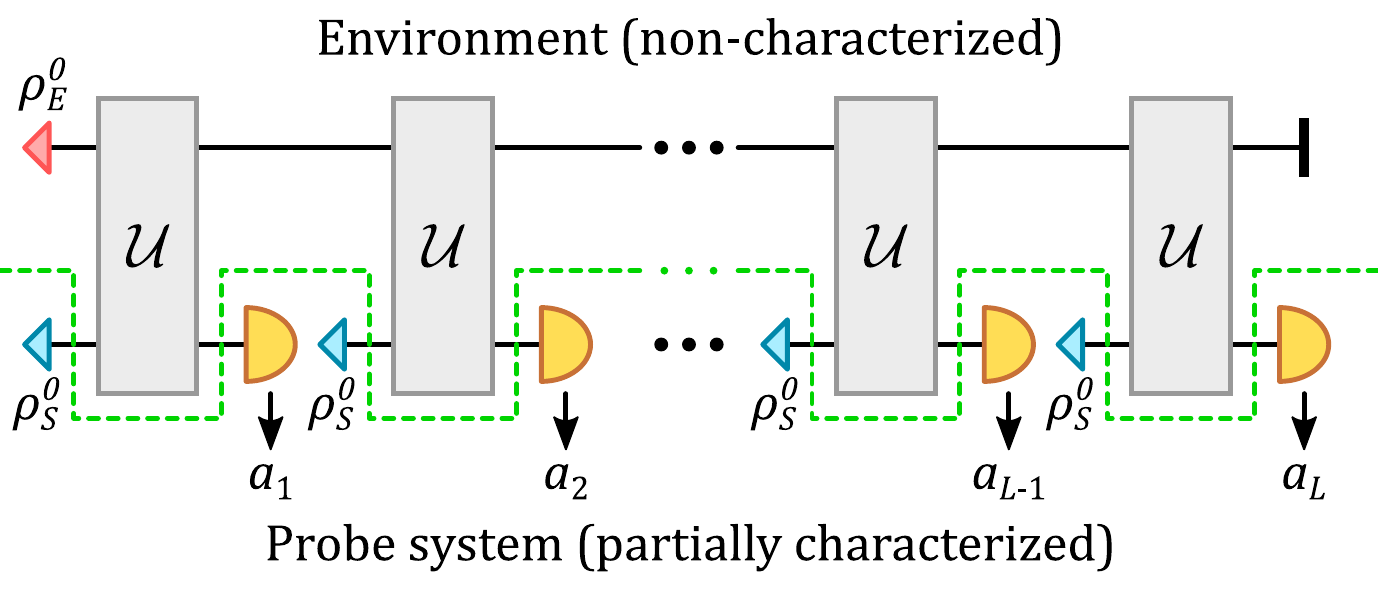"}
	\caption{Diagram of the sequential measurement protocol, involving a partially characterized probe system and a non-characterized environment, here separated by the dashed line. The sequence of measurement outcomes $\va = a_1 \dots a_L$ is obtained by repeated preparations of a probe state $\rhoSz$, fixed and the same at every time step, which is left to interact with the environment, then followed by measurements (semicircles). The environment acts as a \emph{memory} resource, capable of establishing long-term correlations between measurements.}
	\label{fig:protocol}
\end{figure}

We assume throughout that an experimenter is able to prepare the system in a known state and perform measurements in a certain basis. Its dynamics, on the other hand, are not under experimental control, and are governed by its inevitable interaction with the environment. Typically, this environment is a much larger system, generally inaccessible, and featuring complex dynamics. Together, system and environment undergo closed, i.e., unitary evolution. Their interaction leads to an imprinting of information on the probe system, which can be used to learn something about the environment by means of the probe alone. In fact, this is a common way of making indirect measurement on (typically large many-body) systems that are not fully controlled~\cite{braginsky,BuschBook}. For example, a small probe can be used for estimating an unknown parameter of a larger (many-body) environment~\cite{DegenReinhardCappellaroRev17}, in particular temperature~\cite{AlipourMehboudiRezakhaniPRL14,sabin2014impurities,CorreaMehboudiAdessoSanpera15,Mehboudi_2019}. 

We now introduce the details of our scenario and the notation used. By $\cHS$ and $\cHE$ we denote the finite-dimensional Hilbert spaces of system and environment, with $\dS = \dim \cHS$ and $\dE = \dim \cHE$, and their joint space as $\cHES = \cHE \otimes \cHS$, with $\dES = \dE \dS$. 
The experiment involves $L$ identical measurements on the probe system, at discrete time steps, each outcome $a_\ell \in \cA$ collected into a sequence $\va = a_1 \dots a_L$. For simplicity, we consider $\cA = \{0, 1\}$ in the following; the generalization to arbitrary outcomes is straightforward. Measurements on the system are described by a Positive Operator-Valued Measure (POVM) $(E_S^a)_a$, i.e., $E_S^a\geq 0$ and  $\sum_{a\in\cA} E_S^a = \one_S$.

The system is initially in the state $\rhoSz$, interacts with the environment via the global unitary $\cU(\cdot) = U \cdot U^\dagger$, then it is measured and reprepared in the state $\rhoSz$. This procedure is repeated a total of $L$ times; see \cref{fig:protocol}. The unitary operation is assumed to be the same at each iteration (corresponding to, for example the situation of a time-independent system-environment Hamiltonian and temporally equidistant measurements); see \cref{sec:repunitaries} for more details.

In the following, we consider a measure-and-prepare operation of the form
\begin{equation}
	\begin{split}
	\cM_{a}(\rho_{ES}) = \TrP{S}{\rho_{ES} \cdot (\one_E \otimes E^a_S)} \otimes \rhoSz,
	\end{split}
\end{equation}
with the generalization to arbitrary operations being straightforward. The probability of a sequence of outcomes can then be written as
\begin{align}\label{eq:probbasic}
	\begin{split}
	p(\va|\dE) &= p(a_1, \dots, a_L|\dE) \\
	&= \Tr{\cM_{a_L} \!\! \circ \cU \! \circ \dots \circ \cM_{a_1} \!\! \circ \cU(\rhoEz \! \otimes \! \rhoSz)},
	\end{split}
\end{align}
where $p(\va|\dE)$ denotes that it is obtained with an environment of dimension $\dE$. Our task is to establish upper bounds on such probabilities given a finite amount/dimension of ``memory'' $\dE$. The assumption of identical measurements is essential for capturing this notion of memory, as arbitrary time-dependent operations would impose no non-trivial restriction on the probabilities.

For simplicity, we consider the maximization of the probability of a given sequence $\va$, namely, to find $\omega^{\va}_{\dE}$ such that $p(\va|\dE) \le \omega^{\va}_{\dE}$, for all possible correlations generated by an environment of dimension $\dE$. The generalization to arbitrary linear functions of the distributions $(p(\va|\dE))_{\va}$ is straightforward.

\begin{figure}[t]\centering
	\includegraphics[width=1.0\linewidth]{"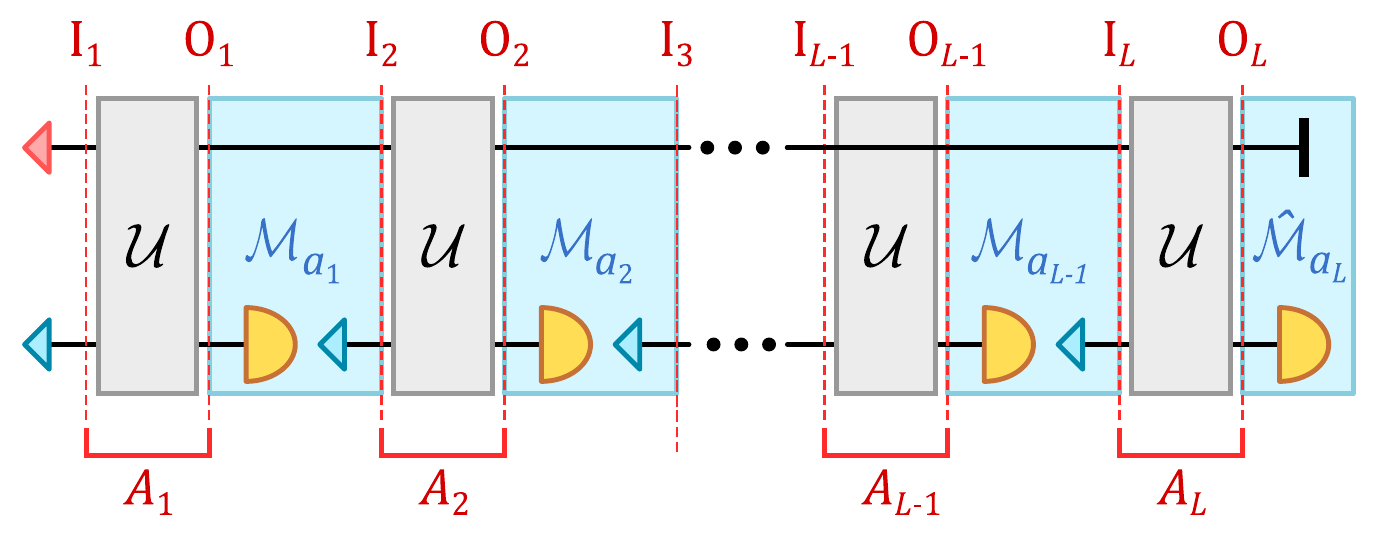"}
	\caption{The protocol with each time step having its own set of input and output Hilbert spaces, denoted by the vertical dotted lines. Note that the input and output spaces of the $\cU$ and $\cM_{a}$ are interleaved: $\cU$ has inputs $\SI_\ell$ and outputs $\SO_\ell$, but for $\cM_{a_\ell}$, inputs are $\SO_\ell$ and outputs $\SI_{\ell+1}$.}
	\label{fig:experimentunraveled}
\end{figure}

We formulate the problem via the Choi-Jamiołkowski (CJ) isomorphism~\cite{jamiolkowski1972,choi1975}, which involves taking multiple copies of the original Hilbert space $\cHES$ to describe its time evolution; see \cref{fig:experimentunraveled}. The $\ell$-th unitary acts as $\cU: \SI_\ell \mapsto \SO_\ell$, while the $\ell$-th measurement, placed between unitaries, acts instead as $\cM_a: \SO_{\ell} \mapsto \SI_{\ell+1}$. All spaces are isomorphic, i.e., $\SO \cong \SI \cong \cHES$. For convenience, the $\ell$-th unitary's input and output spaces are jointly referred to as $\SA{\ell} = \SI_\ell \otimes \SO_\ell$. It is important to emphasize that the local input and output spaces of the $\cU$ and $\cM_{a}$ are interleaved, in the sense that the outputs of one are the inputs of the other (\cref{fig:experimentunraveled}). For clarity, we shall use the cursive $\cI$ and $\cO$ to generically refer to each map's local input and output spaces: for unitaries, $\cI = \SI_\ell$ and $\cO = \SO_\ell$, and for measurements, $\cI = \SO_\ell$ and $\cO = \SI_{\ell+1}$.

Given a linear map $\Lambda: \cI \mapsto \cO$ between input and output Hilbert spaces $\cI$ and $\cO$, with $d_\cI = \dim \cI$ and $d_\cO = \dim \cO$, its Choi-Jamiołkowski representation~\cite{jamiolkowski1972,choi1975} is given by the matrix
\begin{equation}\label{eq:cjiso}
	\Choi{\Lambda} := \sum_{i,j=1}^{d_\cI} \ketbra{i}{j} \otimes \Lambda(\ketbra{i}{j}) = \id_\cI \otimes \Lambda(\purestate{\Omega}),
\end{equation}
where $\ket{\Omega} = \sum_{i=0}^{d_\cI-1} \ket{i i}$ is a non-normalized maximally entangled state. The matrix $\Choi{\Lambda}$ is typically referred to as a Choi matrix for the map $\Lambda$ and, if normalized, also as its Choi state. Note that $\Choi{\Lambda}$ is of size ${d_\cI d_\cO \times d_\cI d_\cO}$. Since in our case input and output spaces are isomorphic, we have $d_\cI = d_\cO = \dES$. We note that the decomposition of sequential maps in \cref{eq:probbasic} is equivalent to the standard formulation of temporal scenarios in the quantum comb~\cite{ChiribellaDArianoPerinottiPRA2009}, process matrix \cite{Oreshkov_2012,ChiribellaDArianoPerinottiValironPRA2013}, and process tensor \cite{PollockPRA2018} formalism.

\section{Witnesses for open system dynamics}\label{sec:opensys}

Given a probe system $S$, the first question we may ask is whether or not it is open, i.e., if it is interacting with an environment at all, or equivalently, whether $\dE > 1$.
To answer this question, we compute $\omega^{\va}_{1} := \max_{\rhoSz, (E_S^a)_a} \; p(\va|\dE=1)$, giving the inequality
\begin{equation}\label{eq:bound_d1}
	p(\va|\dE=1) \leq \omega^{\va}_{1}.
\end{equation}
As an example, the maximum for the sequence $\va = 00101$ is $\omega^{\va}_{1} = (3/5)^3(2/5)^2 = 0.03456$, which we explain how to obtain shortly. If we perform an experiment and observe $p(\va = 00101) = 0.5$, then we must conclude the system is open. In other words, the inequality \cref{eq:bound_d1} acts as a witness for open systems.
However, if no violation is observed, the experiment is inconclusive.

To compute this maximum, we first note that, for our choice of $\cM_a$, since both $\rhoSz$ and $E_S^a$ are the same at each measurement, all outcomes are independent and identically distributed. Writing $q_a = \Tr{\rhoSz E^a_S}$ as the probability of outcome $a$, with $\sum_a q_a = 1$, and using $n_a$ as the number of occurrences of a symbol $a$ in $\va$, we have
\begin{equation}\label{eq:prob_d1}
	p(\va|\dE=1) = \prod_{\ell=1}^{L} q_{a_\ell} = \prod_{a \in \cA} q^{n_a}_a, \qquad \sum_{a \in \cA} n_a = L.
\end{equation}
The global maximum $\omega^{\va}_{1}$ can be found analytically with standard techniques for optimizing $\{q_a\}_a$ (see \cref{app:trivialbound}), and is given by
\begin{equation}\label{eq:trivialbound}
	p(\va|\dE=1) \leq \omega^{\va}_{1} = \prod_{a \in \cA} \left( \frac{n_a}{L} \right)^{n_a},
\end{equation}
with the convention $q^{n_a}_a=1$ for $q_a=n_a=0$. The maximum is independent of $\dS$ and is obtained with $E^a_S = \tfrac{n_a}{L} \one_S$ and any $\rhoSz$. 

Similar principles apply when developing witnesses for the case of $\dE > 1$, although the calculation of the maximum becomes nontrivial, as the probability no longer factorizes as in \cref{eq:prob_d1}. Concretely, we want to obtain the tightest bounds of the form
\begin{equation}\label{eq:oneseqbound}
	p(\va|\dE) \leq \omega^{\va}_{\dE},
\end{equation}
which hold for all possible realizations of the sequence $\va$ in an experiment. If a violation of \cref{eq:oneseqbound} is observed, we can certify that the dimension is greater than $\dE$.

In the simple case of $\dE=1$, we were able to optimize over $\rhoSz$ and $(E_S^a)_a$. In the general case, however, optimizing over all possible protocols (i.e. all possible $\rhoSz$, $(E_S^a)_a$ and $U$) is difficult. Nevertheless, some general considerations can be made about how the bounds depend on these parameters; see \cref{sec:detconds}.
In the following, we assume fixed and well characterized preparations and measurements, in agreement with our initial assumptions on the probe system. The optimization we need to perform is, thus,
\begin{equation}\label{eq:oneseqboundmax}
	\omega^{\va}_{\dE} := \max_{U} p(\va|\dE)
\end{equation}
where assumptions about the initial state of the environment can be removed by convexity and symmetry arguments; see \cref{sec:environment}.
In the following section, we explain how to compute the non-trivial maxima $\omega^{\va}_{\dE}$ via convex optimization methods.

\section{A hierarchy of semidefinite programs bounding temporal correlations}\label{sec:SDPmain}

The maxima in \cref{eq:oneseqbound} require that the probe state $\rhoSz$ and measurements $(E_S^a)_a$ are characterized and fixed throughout the experiment, as discussed in the previous section. To obtain upper bounds for the maximum $\omega_{\dE}^{\va}$, we formulate the problem as a hierarchy of semidefinite programs (SDPs)~\cite{CVXBook}. This class of optimization problems are ubiquitous in quantum information theory, offering strict convergence of solutions for many classes of problems, as well as having efficient algorithms widely available for solving them~\cite{CVXBook,wolkowicz2012sdp}. For a recent compendium on SDPs in quantum information science, see \cite{skrzypczyk2023,tavakoli2023semidefinite}.
In the following, we explain the general steps taken in formulating our problem as an SDP, which can be solved numerically and, importantly, efficiently. A more detailed step-by-step formulation is included in \cref{app:theSDP}.

Under the CJ representation, the repeated applications of the unitary and measurements can each be written as single operators existing in a larger space, encompassing multiple time steps, with their input and output spaces interleaved. This is illustrated in \cref{fig:experimentunraveled}. Via \cref{eq:cjiso}, we define
\begin{equation}\label{eq:mainchoimatrices}
	\begin{split}
	C_U &:= \tfrac{1}{d_{ES}}\cC(\cU), \\
	M_a &:= \cC(\cM_a), \quad\text{and}\quad \hat{M}_a := \TrP{\cO}{M_a}.
	\end{split}
\end{equation}
With this, we can write the $L$ repeated applications of $\cU$ as $C_U^{\otimes{L}}$, and the measure-and-prepare operations as a single operator
\begin{equation}\label{eq:Xdef}
	X := (\dES)^L (\rhoE^0 \otimes \rhoSz) \otimes M_{a_1} \otimes M_{a_2} \otimes \cdots \otimes \hat{M}_{a_L},
\end{equation}
where $\hat{M}_a$ is simply the final measurement without a re-preparation, obtained by partially tracing $M_a$ over its output space. As the goal is to obtain a maximum, the initial environment state $\rhoEz$ can be chosen to be pure by convexity arguments, fixed to be $\ketbra{0}{0}_E$ by unitary invariance, and considered as part of the experimental setup $X$ without loss of generality; see \cref{sec:environment}. We can then express the probability of $\va$ as
\begin{align}
	p(\va|\dE) = \Tr{ X^\T C_U^{\otimes L} }.
\end{align}
Note that by \cref{eq:cjiso} and \cref{eq:mainchoimatrices} $C_U$ is positive, rank-1, and normalized, which makes $C_U^{\otimes{L}}$ a pure symmetric separable quantum state. In addition, each $C_U$ satisfies $\TrP{\cO}{C_U} = \one_{\cI} / \dES$, which arises from the trace-preserving property of each individual unitary.
With all of these observations, we reformulate \cref{eq:oneseqboundmax} as the equivalent optimization problem:
\begin{equation}\label{eq:firstSDPmain}
	\begin{split}
		\textbf{Given:}\quad & \displaystyle \dE,\; \dS,\; \rhoEz,\; \rhoSz,\; \{ M_a \}_a,\; \va, \\
		\text{with} \quad &X := (\dES)^L (\rhoEz \otimes \rhoSz) \otimes M_{a_1} \otimes \cdots \\
		&\cdots \otimes M_{a_{L-1}} \otimes \hat{M}_{a_L}, \\
		\textbf{Find:}\quad & \omega^{\va}_{\dE} := \max_{C_U} \Tr{X^\T C_U^{\otimes L}} \\
		\textbf{Subject to:}\quad & C_U \ge 0, \quad \rank C_U = 1,\\
		&\TrP{\cO}{C_U} = \one_{\cI} / \dES.
	\end{split}
\end{equation}
As $C_U$ enters in the objective function as a tensor power, and we further require it to be rank-1, the problem is both non-linear and non-convex with respect to $C_U$, and thus the above problem cannot be directly solved as an SDP. To reach an SDP \emph{relaxation} of \cref{eq:firstSDPmain}, we first transform it into a chain of equivalent problems.

The problem can be made convex by replacing $C_U^{\otimes L}$ with a separable state on the symmetric subspace
\begin{equation}\label{eq:CULmain}
	\PhitL := \sum_i p_i \purestate{\phi_i}^{\otimes L},
\end{equation}
with $p_i \ge 0$, $\sum_i p_i = 1$.
By convexity, the maximum will be achieved for a rank-1 $\PhitL$, therefore this relaxation leaves the optimal value unchanged. Thus, we replace $C_U^{\otimes L}$ in \cref{eq:firstSDPmain} with $\PhitL$ satisfying
\begin{align}\begin{split}\label{eq:PhitL}
	\PhitL \in \SEP_L, \qquad &\PpL \PhitL = \PhitL,
\end{split}\end{align}
and the partial trace constraint of \cref{eq:firstSDPmain}, where $\SEP_L$ is the set of fully separable $L$-partite states, and $\PpL$ is the projector onto the symmetric subspace of $L$ systems (see the definition in \cref{eq:Pplus} in \cref{app:impsym}). In fact, the constraints in \cref{eq:PhitL} are exact, as it can be shown that all separable states in the symmetric subspace are of the form in \cref{eq:CULmain}; see \cref{subsec::SymRank1} for a proof.

So far, we transformed the original problem into an equivalent one, but due to the condition $\PhitL \in \SEP_L$ it is not yet an SDP. This requirement can be obtained by relaxing the original problem via the quantum de Finetti theorem~\cite{Caves2002} (see also \cite{christandl2007one}), which tells us that
$\PhitL$ can be approximated as the reduced state of a larger symmetric (and potentially entangled) state $\PhiN$, such that $\PhitL \approx \operatorname{tr}_{\AIO{L+1}{N}}[ \PhiN ]$, where $\operatorname{tr}_{\AIO{L+1}{N}}$ is the trace over systems $L+1, \dots, N$.

While this relaxation \emph{a priori} only provides an upper bound for the original problem, it establishes a hierarchy of approximate solutions, which are known to converge exactly to the separable set as $N$ increases~\cite{dps2004completefamily}. The de Finetti theorem, and analogous results for permutationally invariant (or exchangeable), rather than symmetric, operators have found broad applications in quantum information theory, from entanglement detection~\cite{dps2002distinguishing,dps2004completefamily,navascues2009} to more general optimization problems, such as evaluating convex roofs of entanglement measures \cite{TothPRL2015}, constrained bilinear optimization~\cite{BertaMP2021}, dimension-bounded quantum games~\cite{Jee2021}, as well as rank-constrained optimization~\cite{yu2020quantuminspired}.
All the above mentioned results exploit the basic idea of using either the symmetric subspace or permutation invariance to relax nonlinear constraints. In particular, Ref.~\cite{BertaMP2021} analyzed the optimization over pairs of quantum channels, basing their construction on permutation invariant operators, and showing the partial trace constraint necessary to translate the usual procedure (e.g., in the entanglement detection approach) from states to channels. Moreover, \cite{yu2020quantuminspired} introduced the idea of a rank-constrained optimization based on the symmetric subspace, which is central to impose the rank constraint in \cref{eq:firstSDPmain}. Our construction is inspired by these works and could be derived starting from some of these results. However, it is more straightforward to provide a direct construction in terms of the de Finetti theorem; see \cref{app:SDPdeFinetti}.

Ultimately, we arrive at an SDP relaxation of the original problem:
\begin{equation}\label{eq:finalSDPmain}
	\begin{split}
		\textbf{Given:}\quad & \displaystyle \dE,\; \dS,\; \rhoEz,\; \rhoSz,\; \{ M_a \}_a,\; \va,\; N, \\
		\text{with} \quad &X := (\dES)^L (\rhoEz \otimes \rhoSz) \otimes M_{a_1} \otimes \cdots \\
		&\cdots \otimes M_{a_{L-1}} \otimes \hat{M}_{a_L}, \\
		\text{and} \quad &\hat{M}_{a_L} := \TrP{\cO}{M_{a_L}} \\
		\textbf{Find:}\quad & \tilde{\omega}^{\va,N}_{\dE} := \max_{\PhiN} \Tr{(X^\T \otimes \one_{\AIO{L+1}{N}}) \PhiN} \\
		\textbf{Subject to:}\quad & \PhiN \ge 0,\\
		&\PpN \PhiN = \PhiN, \;\; \Tr{\PhiN} = 1, \\
		&\TrP{\SO_1}{\PhiN} = \frac{\one_{\SI_1}}{\dES} \otimes \TrP{\SA{1}}{\PhiN}.
	\end{split}
\end{equation}
Here, $N \ge L$ is the size of the symmetric state $\PhiN$ used in the approximation of the separable $L$-partite state $\PhiL$. The last constraint in \cref{eq:finalSDPmain} is required to ensure that $\PhiN$ represents a sequence of trace-preserving maps, and it can be enforced on a single step due to the symmetry constraint on ${\PhiN}$; see \cref{sec:TPcond}. Optionally, one can also add further linear constraints to the SDP to improve its approximation to the separable set, e.g., entanglement witnesses~\cite{terhal2000,guhne2009} or the positive partial transpose (PPT) criterion~\cite{peres1996ppt} ${\PhiN^{\Ta} \ge 0, \forall \alpha}$, where $\PhiN^{\Ta}$ is the partial transpose with respect to a bipartition $\alpha$. These conditions are not necessary for convergence, but they may provide better results~\cite{navascues2009} at an extra computational cost.

Once approximate solutions $\tilde{\omega}^{\va,N}_{\dE}$ are obtained from \cref{eq:finalSDPmain}, they can be used to establish a convergent hierarchy of upper bounds on the temporal correlations for each sequence, i.e.,
\begin{align}\label{eq:hierarchybounds}
	\begin{split}
	&\omega^{\va}_{\dE} \le \cdots \le \tilde{\omega}^{\va,N+1}_{\dE} \le \tilde{\omega}^{\va,N}_{\dE} \le \cdots \le \tilde{\omega}^{\va,L}_{\dE}, \\
	&\text{with}\quad \lim_{N \to \infty} \tilde{\omega}^{\va,N}_{\dE} = \omega^{\va}_{\dE}.
	\end{split}
\end{align}
Consequently, the above formulates a sequence of SDPs capable of approximating the maxima $\omega_{\dE}^{\va}$ to arbitrary precision. Numerical results obtained from \emph{any} of these SDPs (i.e., for any $N \geq L$) can be used to construct dimensional witnesses, which, when violated, certify the minimum dimension of the effective environment interacting with the system.
A software implementation of the SDP in \cref{eq:finalSDPmain}, however, is not straightforward even for small values of $\{\dS, \dE, L, N\}$, as the memory requirements quickly render the problem computationally intractable. Obtaining the numerical results presented in the next section, thus, required a significant amount of optimization; see \cref{sec:implementation} for details. A schematic outline of all steps undertaken to formulate the SDP and obtain the numerical results is presented in \cref{fig:relaxations}.

\begin{figure*}[t]\centering
	\includegraphics[width=0.65\linewidth]{"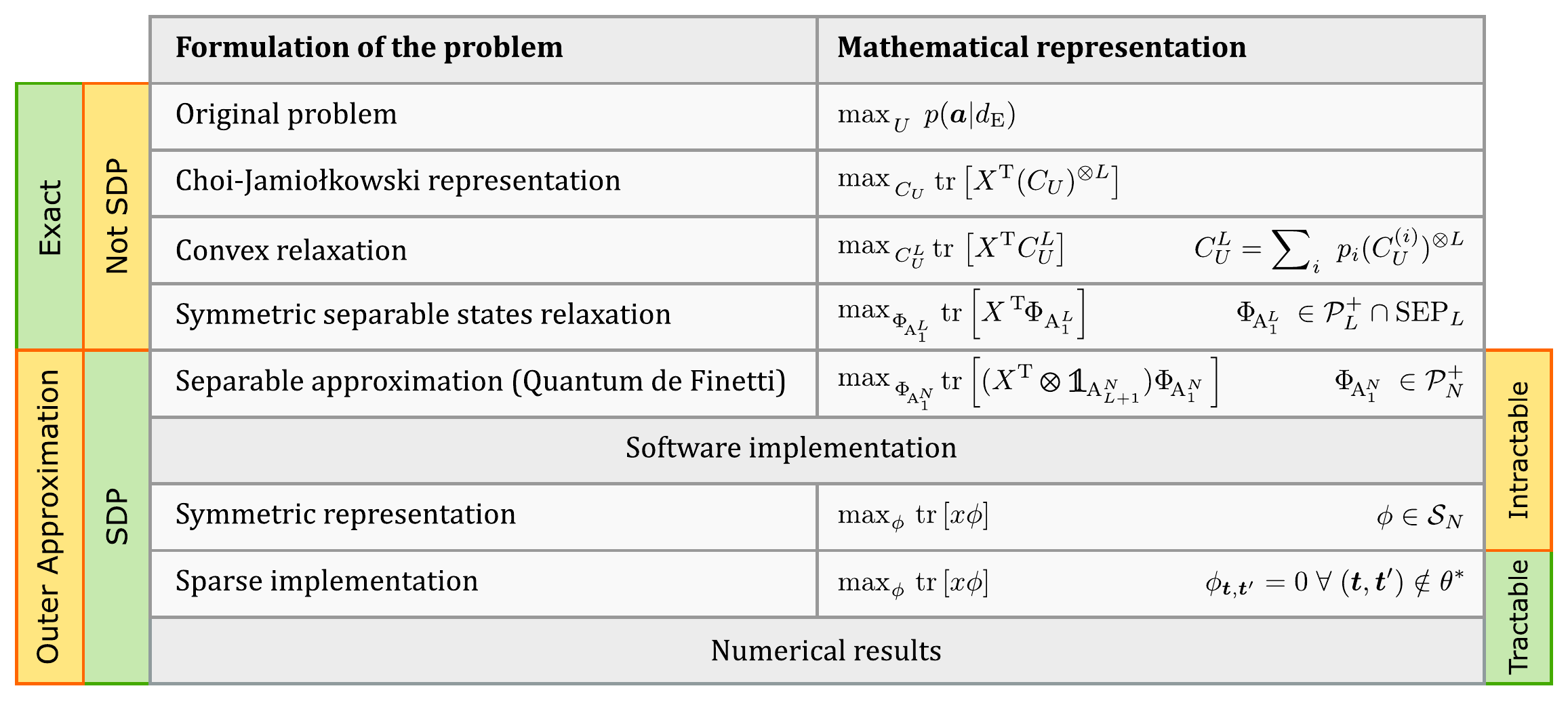"}
	\caption{Schematic of all steps undertaken for computing an upper bound for $\omega_{\dE}^{\va}$ by formulating and solving the SDP problem. Detailed descriptions of each step are covered in \cref{app:theSDP,app:implementation}. The tractable/intractable labels, on the right, refer to the case $\dE = 2$.}
	\label{fig:relaxations}
\end{figure*}

The asymptotic convergence of the hierarchy of bounds in \cref{eq:hierarchybounds} for the SDP in \cref{eq:finalSDPmain}, without PPT constraints, is given by~\cite{chiribella2011}
\begin{equation}
	|\omega^{\va}_{\dE} - \tilde{\omega}^{\va,N}_{\dE}| \le \frac{L(L+(\dES)^2+1)}{N + \dES}.
\end{equation}
If PPT constraints are included, we only have partial analytical results for the asymptotic error. Numerical evidence and previous results involving PPT constraints~\cite{navascues2009} lead us to conjecture an asymptotic scaling of the form
\begin{equation}
	|\omega^{\va}_{\dE} - \tilde{\omega}^{\va,N}_{\dE}| \le f(L,\dES) \, O\left(\frac{1}{N^2}\right),
\end{equation}
where $f(L,\dES)$ is a function of $L$ and $\dE$. We leave the analysis of these asymptotic error bounds to a future investigation.

\section{Numerical results}\label{sec:numerics}

In this section, we discuss the numerical results we obtained for $\tilde{\omega}^{\va,N}_{\dE}$. The SDPs were run with CVXPY~\cite{diamond2016cvxpy,agrawal2018rewriting} using the solver SCS~\cite{scspaper,scssoft}, on a compute server with an Intel Xeon Gold 5218 24-core processor at 2.294 GHz, and with 128 GB of RAM. For $\dE = 1$, optimization time was in the order of seconds, whereas $\dE = 2$ scenarios required from 4 up to 14 hours to complete.

For the explicit computation, we chose the following measurement protocol: $\cA = \{0,1\}$, $\rhoEz = \ketbra{0}{0}_E$,  $\dS = 2$, $\rhoSz = \ketbra{0}{0}_S$ and $E_S^a = \ketbra{a}{a}_S$, so that
\begin{equation}\label{eq:projMa}
	\cM_a(\rhoES) = \Tr{ \rhoES \cdot \one_E \otimes \ketbra{a}{a}_S} \otimes \ketbra{0}{0}_S.
\end{equation}
For this particular choice, the bounds $\omega^{\va}_{\dE}$ are the same as for the isolated system case studied in \cite{vieira2022temporal}, where explicit time evolutions were found through gradient descent techniques.

The analytical maxima for $\dE = 1$ described in \cref{sec:opensys} served as a useful test for the soundness of our approach. We ran the SDP for all binary sequences (up to $0 \leftrightarrow 1$ relabeling symmetry), for $L = 2, 3, 4$, $N = L, L+1$. Every case was tested with and without the additional PPT constraints for comparison. Results for $L = 2, 3$ are shown in \cref{tbl:trivialSDPresults}, in which it can be observed that, for $\dE = 1$, either the PPT constraints or a symmetric extension of a single extra system appear to be sufficient for achieving the exact analytical maximum.

\begin{table}
	\renewcommand{\arraystretch}{1.2}
	\centering
	\begin{equation*}\footnotesize
		\begin{array}{c|c|l|c|c|c}\hline
			L & N & \va & \text{Without PPT} & \text{With PPT} & \omega^{\va}_{1} \\
			\hline
			2 & 2 & 00  & 0.999996 & 0.999999 & 1 \\
			2 & 3 & 00  & 1.000000 & 1.000000 & 1 \\ \hline
			2 & 2 & 01  & 0.500000 & 0.250002 & 1/4 \\
			2 & 3 & 01  & 0.250007 & 0.250001 & 1/4 \\ \hline
			3 & 3 & 000 & 1.000010 & 1.000000 & 1 \\
			3 & 4 & 000 & 1.000000 & 0.999991 & 1 \\ \hline
			3 & 3 & 001 & 0.250005 & 0.148149 & 4/27 \\
			3 & 4 & 001 & 0.148148 & 0.148148 & 4/27 \\ \hline
			3 & 3 & 010 & 0.250005 & 0.148149 & 4/27 \\
			3 & 4 & 010 & 0.148148 & 0.148148 & 4/27 \\ \hline
			3 & 3 & 011 & 0.250005 & 0.148149 & 4/27 \\
			3 & 4 & 011 & 0.148148 & 0.148148 & 4/27 \\ \hline
			\hline
		\end{array}
	\end{equation*}
	\vspace{-1.5em}
	\caption{Results of the SDP for $\dE = 1$, compared with the analytical maxima $\omega^{\va}_{1}$. Note that $4/27 = 0.\overline{148}$. It can be seen that either PPT constraints or an extension of one extra system (i.e. $N = L+1$) was sufficient to achieve the analytical maximum in this case. Similar results were obtained for $L=4$, omitted here for conciseness.}
	\label{tbl:trivialSDPresults}
\end{table}

We have also solved the SDP for the case $\dS =\dE = 2$, for various sequences without a trivial maximum probability (i.e., $\omega^{\va}_{\dE} < 1$), but only without the additional PPT constraint, as the requirements needed for its addition would have exceeded the available memory. Additionally, since the involved SDPs are only numerically tractable for $N \approx L$, and no convergence was observed for the values we managed to compute, we can only claim to have obtained upper bounds for the maximum $\omega_{\dE}^{\va}$, as in \cref{eq:hierarchybounds}.

\begin{table}
	\renewcommand{\arraystretch}{1.2}
	\centering
	\begin{equation*}\footnotesize
	\begin{array}{c|c|l|c|c}\hline
	L & N & \va & \text{Upper bound (SDP)} & \text{Achievable value (GD)} \\
	\hline
	3 & 3 & 001  & 0.683477 & 0.437341 \\
	3 & 4 & 001  & 0.521219 & 0.437341 \\
	\hline
	4 & 4 & 0010 & 0.512220 & 0.437142 \\
	\hline
	4 & 4 & 0011 & 0.487058 & 0.362047 \\
	\hline
	4 & 4 & 0100 & 0.494499 & 0.333147 \\
	\hline
	4 & 4 & 0110 & 0.488837 & 0.361968 \\
	\hline
	4 & 4 & 0001 & 0.492088 & 0.300545 \\
	\hline
	\hline
	\end{array}
	\end{equation*}
	\vspace{-1.5em}
	\caption{Comparison between SDP upper bounds and best achievable values known for the maximum $\omega_{\dE}^{\va}$, for $\dE = 2$. Best known values were obtained through gradient descent (GD)~\cite{vieira2022temporal} and act as lower bounds on the maximum, but are not suitable as witnesses. All results are in agreement, and for $\va = 001$, we observe convergence of the bounds as $N$ increases.}
	\label{tbl:nontrivialSDPresults}
\end{table}

As our choice of probe state and measurements (\cref{eq:projMa}) for $\dS = \dE = 2$ reproduces the quantum scenarios previously investigated in \cite{vieira2022temporal}, as we show in \cref{sec:detconds}, we can compare our upper bounds against the largest known achievable values, obtained through explicit time-evolutions found via gradient descent techniques. This allows us to estimate the range of values containing the maxima $\omega_{2}^{\va}$. These comparisons are shown in \cref{tbl:nontrivialSDPresults}, where, for instance, we see $0.437341 \le \omega_{2}^{001} \le 0.521219$. As a concrete example of a violation of this bound, for $\dS = 2$ and $\dE = 3$, we may construct a unitary in $\cHES$ as follows:
\begin{align}\begin{split}\label{eq:detU}
	U &= \ketbra{1}{0}_E \!\otimes\! \ketbra{0}{0}_S + \ketbra{2}{1}_E \!\otimes\! \ketbra{0}{0}_S +  \ketbra{2}{0}_E \!\otimes\! \ketbra{1}{1}_S \\
	&+ \ketbra{1}{1}_E \!\otimes\! \ketbra{1}{1}_S + \ketbra{0}{2}_E \!\otimes\! \ketbra{0}{1}_S + \ketbra{0}{2}_E \!\otimes\! \ketbra{1}{0}_S.
\end{split}\end{align}
For the $\rhoEz$, $\rhoSz$ and $E_S^a$ we have chosen, the above unitary gives $p(001|\dE=3) = 1 > \omega_{2}^{001}$, implying that the bound found for $\dE = 2$ is not only non-trivial, but actually witnesses environment dimension since it can indeed be violated by means of a larger environment ($\dE=3$ in this case). In fact, this example certifies that $\omega_{3}^{001} = 1$; see also  \cref{sec:detconds}.

\section{Implementation}\label{sec:implementation}

This section outlines the concrete steps we have taken to relax the original problem to a hierarchy of SDPs and to make these SDPs numerically tractable; see \cref{app:implementation} for more details. As noted previously, a straightforward implementation is not computationally tractable, even for short sequences, due to the large number of variables and constraints involved.

We have addressed these additional numerical challenges by exploiting several properties of the problem. Firstly, the symmetric constraint imposed on $\PhiN$ can be satisfied automatically by expressing $\PhiN$ in terms of a basis for the symmetric subspace in the numerical implementation of the SDP. While this provides a significant reduction in the total number of involved variables and constraints, it is by itself insufficient to make the problem tractable.

However, thanks to our specific choice of initial states and measurements, the sparsity of the objective function could be exploited to further reduce the number of variables and constraints, by eliminating all of those that do not affect the objective, either directly or indirectly. This relaxation generally results in a sparse outer approximation of the original problem, i.e., to a value $\Omega^{\va,N}_{\dE} \geq \tilde{\omega}^{\va,N}_{\dE}$. In our particular case, as we show in \cref{app:impsparse}, this elimination procedure is exact, thus yielding an optimal value $\Omega^{\va,N}_{\dE} = \tilde{\omega}^{\va,N}_{\dE}$.

\subsection{Choice of parameters}

Before a concrete implementation, one must first choose the initial environment state $\rhoEz$, as well as fixing the probe state and measurements. As the objective function is convex on $\rhoEz$, and the optimization is over all unitaries, we may -- without loss of generality -- fix a pure initial state $\rhoEz = \ketbra{0}{0}_E$ and eliminate any assumptions on the environment state; see \cref{sec:environment}.
On the other hand, as assumptions on $\dS$, $\rhoSz$ and $(E_S^a)_a$ are experiment-dependent, one must search for ``proper'' probe states and measurements on a case-by-case basis, but optimal choices, leading to a larger bound, can also be addressed in general terms; see \cref{sec:detconds}. For the choice of sequence, in particular, if it is too simple relative to $\dE$, the maximum attainable probability may be trivial (i.e., equal to one), such that no optimization is required, and no witness can be constructed. Therefore, it is important to choose sequences which have a non-trivial maximum for a given $\dE$. For closed systems, this problem has been solved via the notion of ``deterministic complexity'' \cite{vieira2022temporal}, which defines the minimum requirements for a sequence to be able to occur deterministically. In the open system case, the conditions for determinism involve not only the available memory, as was the case in the isolated system, but also the dimension of the system. We elaborate on these conditions in \cref{sec:detconds}.

\subsection{Symmetric representation}

Effectively solving any of the SDPs in the hierarchy requires a significant amount of optimization, as a naive implementation quickly becomes computationally intractable. As a rough example, without any simplifications, $\PhiN$ is a square matrix of size $(\dES^2)^N$, which in the simplest non-trivial scenario, $\dE = \dS = 2$ and $N = L = 3$, already results in a $4096 \times 4096$ matrix, with over 16 million complex scalar variables. The partial trace constraint alone involves over 1 million linear equations between these variables, making the SDP numerically intractable in this na\"ive formulation.

In practice, several simplifications can be made; see \cref{app:implementation} for full details. $\PhiN$ and $X \otimes \one_{\AIO{L+1}{N}}$ may be written directly in terms of an operator basis for the symmetric subspace, thus automatically satisfying the symmetry constraint, which significantly reduces the number of variables. Defining an isometry between the symmetric subspace of $\AIO{1}{N}$, denoted by $\cS_N$ (with dimension $\dim \cS_N$), and a canonical basis for $\cS_N$,
\begin{equation}
	V = \sum_{\vt} \ket{\vt}_{\cS_N} \bra{\text{sym}(\vt)}_{\AIO{1}{N}},
\end{equation}
we may cast all involved objects directly as ${\dim \cS_N \times \dim \cS_N}$ matrices:
\begin{align}\label{eq:symversion}
	\begin{split}
	\PhiN &= \sum_{\vt,\vt'} \phi_{\vt,\vt'} \ketbra{\text{sym}(\vt)}{\text{sym}(\vt')}_{\AIO{1}{N}}, \\
	\phi &= V \PhiN V^\dagger, \; \text{and} \\
	x &= V (X^\T \otimes \one_{\AIO{L+1}{N}}) V^\dagger,
	\end{split}
\end{align}
where $\vt$ denotes a ``type'' for the canonical representation of symmetric states $\ket{\text{sym}(\vt)}_{\AIO{1}{N}}$; see \cite{harrow2013church}. Under this representation, the objective function becomes $\Tr{x \phi}$, and the constraints may be formulated in terms of $\phi_{\vt,\vt'}$ directly, without resorting to the the full space $\AIO{1}{N}$. Note that $V V^\dagger = \one_{\cS_N}$ but $V^\dagger V \ne \one_{\AIO{1}{N}}$ (rather, $V^\dagger V$ is a projector onto $\cS_N$). As the dimension of $\cS_N$ is~\cite{harrow2013church}
\begin{equation}
	\dim \cS_N = \binom{N + \dES^2 - 1}{N},
\end{equation}
for $\dES=4$ and $N=3$ the matrix $\phi$ is of size $816 \times 816$, with around 600 thousand variables: a reduction to 4\% of the original 16 million. With further work, the partial trace constraints can also be efficiently written directly in this representation, leading to $( \dES \cdot \binomial{(N-1) + \dES^2 - 1}{N-1} )^2$ equality constraints, resulting in approximately 300 thousand constraints in the $\dES=4$ and $N=3$ scenario: a reduction to 28\% of the original. This symmetric representation is explained in more detail in \cref{app:impsym}.

\subsection{Sparse implementation}

While these are significant improvements, they still prove to be insufficient, as the dense representation of the symmetric problem involves too many variables and constraints to be solvable. As a concrete example, for the simplest case of $\va = 001$, $\dE=2$ and $L=N=3$, the SDP solver SCS attempts to allocate a dense array of over 3 TB in size.

To address this, we have developed an algorithm that exploits any sparsity of the objective function, and -- whenever possible -- constructs a sparse relaxation of the original problem. This is achieved by iteratively selecting which variables of $\phi$ and which of its constraints are strictly necessary to solve the problem, due to their direct or indirect influence in the objective function. While generally a relaxation, due to the explicit structure of our problem this sparse implementation is exact, yielding $\tilde{\omega}^{\va,N}_{\dE}$. A detailed explanation of our algorithm is available in \cref{app:impsparse}.

For the state and measurements considered, our technique was capable of reducing the number of variables and constraints immensely, to less than 1\% of the symmetric case (see \cref{tbl:symsparsecomparison}). This allowed us to successfully compute upper bounds for various sequences up to $N = 4$ and $\dE = 2$ (\cref{tbl:nontrivialSDPresults}).

\begin{table}\footnotesize
	\renewcommand{\arraystretch}{1.2}
	\centering
	\begin{tabular}{c|c|r|r|c} \hline
		$N$                  &                      & Symmetric & Sparse & Reduction to \\ \hline
		\multirow{2}{*}{3} & \multicolumn{1}{l|}{\textbf{Variables}}  & 665\,856                         & 3\,566                        & 0.54\%    \\
		& \multicolumn{1}{l|}{\textbf{Constraints}} & 295\,937                         & 2\,809                        & 0.95\%    \\ \hline
		\multirow{2}{*}{4} & \multicolumn{1}{l|}{\textbf{Variables}}  & 15\,023\,376                       & 35\,688                       & 0.24\%    \\
		& \multicolumn{1}{l|}{\textbf{Constraints}} & 10\,653\,697                       & 40\,441                       & 0.38\% \\
		\hline \hline
	\end{tabular}
	\caption{Comparison between the number of variables and linear constraints in the symmetric problem vs. its sparse implementation, for the sequence $\va = 001$. Our algorithm achieves a vast reduction in the number of variables and constraints, while still allowing exact solutions for our problem. Note that these numbers precede any further optimization, which could remove any remaining redundant constraints or variables.}
	\label{tbl:symsparsecomparison}
\end{table}

As can be observed in \cref{tbl:symsparsecomparison}, the sparse problems can clearly be simplified further, by eliminating redundant variables and constraints. We opted for leaving such task to the numerical pre-solver, as this only took a few minutes of computing time. Solving the SDP for $N = 3$ was possible within minutes, but for $N = 4$, from 4 up to 14 hours were needed, depending on the sequence.

\section{Discussion}\label{sec:discussion}

\subsection{Repeated unitaries}\label{sec:repunitaries}

In the following, we discuss in which cases the condition of repeated unitaries is physically justified. To see this, imagine that the evolution of the system is governed by a time-independent Hamiltonian $H_{\rm{SE}}$. The corresponding unitaries will be of the form $U(t) = e^{-iHt}$. This assumption of time-independence is always possible, as the environment may include any source of time dependence. Ideally, then, we would have the same unitary if we perform the measure-and-prepare operations equally spaced in time, say, by an interval $\Delta t$. One should take into account, however, some uncertainty in the time measurement. Let us then assume that our choice of time for performing a measurement is distributed according to a distribution $q(t)$, centered around $\Delta t$. This means that instead of transforming our system according to the unitary $U(\Delta t)$, we transform it, on average, according to the mapping $\Lambda := \int \cU(t) q(t) dt$, seemingly breaking our assumption of repeated unitaries. Since $\Lambda$ is a valid completely positive trace preserving (CPTP) map, it can be dilated by means of a larger environment into a unitary $\cU'$, namely $\Lambda(\rho_{\rm SE})=\TrP{{\rm E}_1}{\cU'(\rho_{\rm SE} \otimes \ketbra{0}{0}_{{\rm E}_1})}$. To complete the argument that this situation can still be described by repeated unitaries, it is enough to show that the same can be done for multiple copies of it, namely, that the repeated operation $\Lambda^L$ can be dilated to some unitary $\widetilde\cU^L$. To do so, it is enough to provide at each time-step $i$ a new environment ${\rm E}_i $ prepared in the correct initial state $\ketbra{0}{0}_{{\rm E}_i}$. Let us compute explicitly the case for two time steps, the general case is straightforward. We define $W_{\rm{E}_1\rm{E}_2}$ the swap operator between systems ${\rm E}_1$ and ${\rm E}_2$, and we set $\widetilde{\cU}:= W_{\rm{E}_1\rm{E}_2} \circ \, \cU_{\rm{SEE}_1}'$. Let us calculate, for simplicity
\begin{equation}
\begin{split}
\TrP{\rm{E}_1 \rm{E}_2}{ \cU_{\rm{SEE}_1}' \!\!\circ\! W_{\rm{E}_1\rm{E}_2} \circ \cU_{\rm{SEE}_1}'(\rho_{\rm{SE}} \!\otimes\! \ketbra{0}{0}_{\rm{E}_1} \!\otimes\! \ketbra{0}{0}_{\rm{E}_2}) }\\
=\TrP{\rm{E}_1 \rm{E}_2}{  \cU_{\rm{SEE}_1}' \!\!\circ\! W_{\rm{E}_1\rm{E}_2} (\sigma_{_{\rm{SEE}_1}} \!\otimes\! \ketbra{0}{0}_{\rm{E}_2}) }\\
=\TrP{\rm{E}_1 }{  \cU_{\rm{SEE}_1}'( \TrP{\rm{E}_2}{\sigma_{_{\rm{SEE}_2}}} \!\otimes\! \ketbra{0}{0}_{\rm{E}_1}) }\\
=\TrP{\rm{E}_1 }{  \cU_{\rm{SEE}_1}' (\Lambda(\rho_{\rm{SE}})\otimes \ketbra{0}{0}_{\rm{E}_1}) }=\Lambda^2(\rho_{\rm{SE}}),
\end{split}
\end{equation}
where $\sigma_{\rm{SEE}_1}:=\cU_{\rm{SEE}_1}'(\rho_{\rm{SE}}\otimes \ketbra{0}{0}_{\rm{E}_1})$. The effect of the final $W_{\rm{E}_1\rm{E}_2}$ operation is to swap the space $\rm{E}_1$ and $\rm{E}_2$ that are irrelevant after the operations on $\rm{SE}$ have been performed. Evidently, this argument can straightforwardly be extended to more time steps.
In summary, this means that the condition of ``the same unitary'' is satisfied as long as the time choice is always drawn from the same distribution. In other words, it is sufficient to always ``probabilistically repeat'' the same operation.

Finally, we remark that even if this procedure may use a very large environment, we are only interested in showing that we can always assume there \emph{is} a unitary dynamics, with the notion of an effective environment, then, taking care of estimating the environment consistent with the observed statistics. Moreover, we also remark that time-independent operations are necessary, as unrestricted time-dependent operations can achieve arbitrarily long temporal correlations even with bounded memory size~\cite{mao2022strucdimbound}.

\subsection{Effective environment and initial state}\label{sec:environment}

The environment of a quantum system, intended as all the physical systems surrounding and possibly interacting with it, is typically a very high-dimensional system, if not directly assumed to be infinite-dimensional. From this perspective, we want to make sense of the notion of \emph{effective} environment. Consider a global transformation of the system and environment. As a first approximation, we can say that the global unitary is of the form $U_{\text{SE}}\otimes U_{\text{E}'}$, where $U_{\text{SE}}$ is an entangling unitary between the system and the effective environment and $U_{\text{E}'}$ is acting on the rest of the environment.
At the same time, an evolution of the form $U_{\text{SE}}\otimes U_{\text{E}'}$ is just an approximation of the full evolution of the environment, as we expect its state to thermalize after some time interval. Nevertheless, from a physical perspective this approximation is still valid if the time required for a single run of the experiment, i.e., the measurement of the temporal sequence, is much shorter than the time needed for the effective environment to thermalize. This is to be expected, as the environment is composed of many particles that may interact with each other with different strengths. Underlying this expectation is the assumption -- known as Markovian embedding~\cite{budini_embedding_2013, xue_quantum_2015, xue_modeling_2020} and frequently employed in the description and modeling of open quantum system dynamics~\cite{tamascelli_nonperturbative_2018, luchnikov_dimension_2018, luchnikov_simulation_2019, luchnikov_probing_2022} -- that the environment can always be split into two parts: a far environment, leading to irretrievable, memoryless information loss, and an effective environment, that can transport memory. Due to the irretrievable information loss, the dynamics of the system and the effective environment is then described by a general (non-unitary) CPTP map, a situation which, as discussed in the previous section, can again be dilated to repeated unitaries.

Reversing this perspective, namely, looking at the problem of characterizing the effective environment from temporal correlation experiments, we may say that this characterization is inherently dependent on the typical time-scales of a single experimental run.
This difference in time scales and the eventual thermalization of the environment, however, are essential for repeating the experiment for collecting statistics, as it is required that the initial state of the environment is always (probabilistically) described by the same state $\rhoEz$ in each experimental run. As previously mentioned, this is typically a thermal state, but it does not need to be characterized in our approach. In fact, since the SDP maximizes the temporal correlations, which are linear in the initial state of the environment, we know that the maximum is always achieved with a pure state. Up to local unitaries, we can thus always assume it to be $\rhoEz = \ketbra{0}{0}_E$. The bound calculated for this state is, then, valid for any possible initial state of the environment.

Finally, somewhat independent of the explicit experimental situation, we may see our setup as a question of simulation resources: What is the smallest dimension an environment coupling to a known system must have in order to reproduce observed statistics in a unitary way? Seen in this way, the results we present attribute a ``simulation hardness'' (in the sense of required environment dimension) to each sequence, that is agnostic with respect to concrete time scales or experimental limitations, but rather inherent to the respective sequence.

\subsection{Conditions for deterministic realizations}\label{sec:detconds}

The maximum probability for a given sequence increases as more memory is available, i.e.,
\begin{equation}
0 \le \omega_{1}^{\va} \le \omega_{2}^{\va} \le \cdots \le \omega_{\dE}^{\va} \le 1,
\end{equation}
as larger environments can always simulate the dynamics of smaller ones. Therefore, if the sequence $\va$ is sufficiently simple, its maximum may be trivial for a given $\dE$, i.e., $\omega^{\va}_{\dE} = 1$.

The maxima $\omega^{\va}_{\dE}$ depend on the sequence, as not all sequences are equally ``difficult'' to produce with a given amount of memory $\dE$. For example, it is easy to see that the sequence `$000$' can \emph{always} be produced with unit probability, while the sequence `$001$' cannot be produced with unit probability if the environment is not of sufficient size~\cite{vieira2022temporal}. This observation suggests a relevant notion of ``complexity'' of sequences, which offers a more fundamental relationship between $\va$ and $\dE$. Since our goal is to establish non-trivial maxima on the probabilities of sequences, the natural question to ask is how large should $\dE$ be such that $\va$ can occur deterministically?

Understanding the conditions where this occurs allows us to pick only scenarios featuring non-trivial maxima. It is useful to recall the following
\begin{definition}(Deterministic Complexity \cite{vieira2022temporal}) Let $\va$ be a sequence of measurement outcomes, produced by repeated identical measurements on an isolated $d$-dimensional system. The \emph{deterministic complexity} of the sequence $\va$, or $\DC(\va)$, is the minimum $d$ such that $p(\va|d) = 1$ can occur.
\end{definition}
$\DC(\va)$ is a property of the sequences $\va$ and it is independent of the model (quantum or classical). 
\begin{figure}\centering
	\includegraphics[width=0.85\linewidth]{"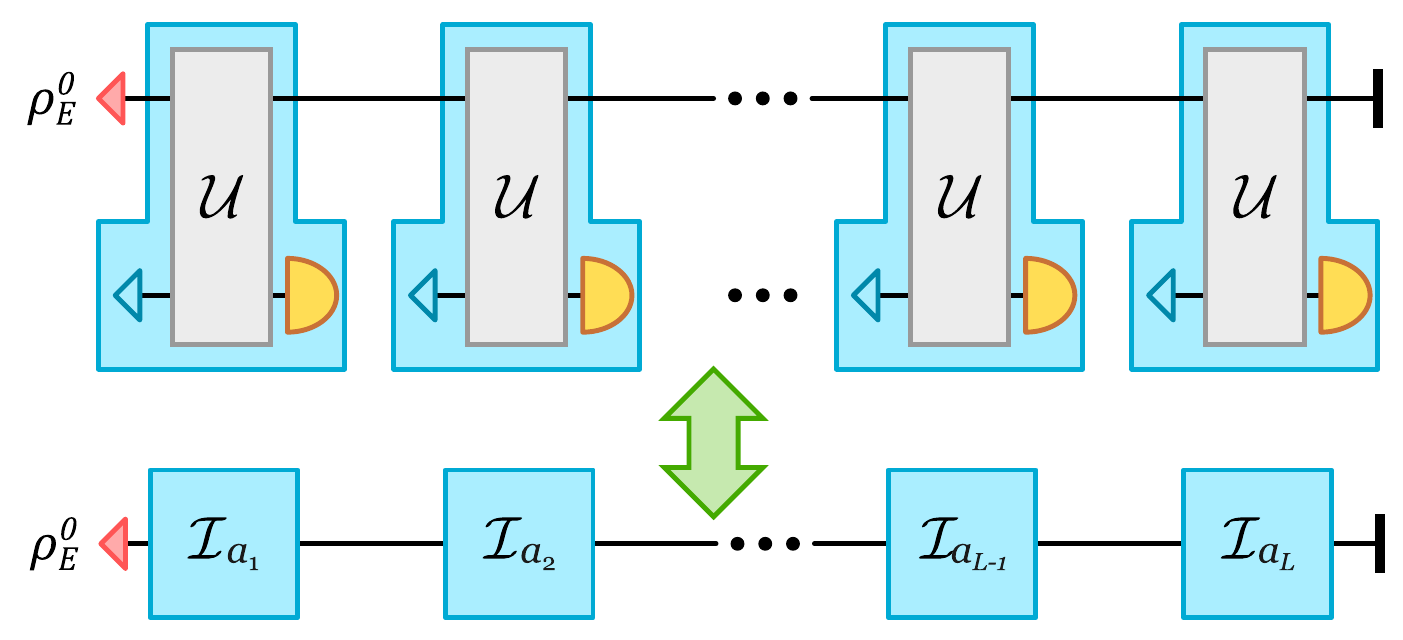"}
	\caption{The open system and environment scenario can be equivalently interpreted as the dilation of the scenario involving sequential measurements $\cI_a$ on an isolated environment system.}
	\label{fig:protocoldilation}
\end{figure}
First, we establish a correspondence between the single system and open system scenarios. Let $\dS = |\cA|$, so that we may choose $E_S^a = \ketbra{a}{a}_S$ and $\rhoSz = \ketbra{0}{0}_S$. Then, the changes in the environment state due to the unitary and measurements can be written succinctly in terms of the completely positive maps
\begin{align}\label{eq:DCIa}
	\cI_a(\rhoE) &= \TrP{S}{U (\rhoE \!\otimes\! \ketbra{0}{0}_S) U^\dagger \!\cdot\! \one_E \otimes \ketbra{a}{a}_S}.
\end{align}
Therefore, under the action of $(\cI_a)_{a \in \cA}$, we may interpret the \emph{environment} as a single isolated ``memory'' system in which the maps $\cI_a$ act sequentially, as in \cite{vieira2022temporal}. The converse mapping, from the sequential measurements on an isolated memory system to the open system scenario, is the dilation map with the probe system acting as the ancilla and the environment as the memory. This construction is well known, but it is worth expanding it in detail in order to understand what are the conditions on the ancilla (i.e., the probe system) needed to implement it.

Let $n(\va)$ denote the number of unique symbols appearing in $\va$, with $n(\va) \le |\cA|$, and let us assume $\dS \geq n(\va)$. Now, consider an instrument given by the maps $\tilde{\cI}_a$ of the form $\tilde{\cI}_a(\rhoE) = K_a \rhoE K_a^\dagger$, for $(K_a)_{a\in\cA}$ Kraus operators, i.e., $\sum_{a \in \cA} K_a^\dagger {K_a} = \one_E$. The corresponding unitary arises from the $\dE \dS \times \dE$ isometry matrix
\begin{equation}
	Q = \sum_{a \in \cA} K_a \otimes \ket{a}_S,
\end{equation}
which can always be completed into a unitary matrix $U$. Thus $\tilde{\cI}_a$ can be written in the form of \cref{eq:DCIa}. As we have chosen $\dS = |\cA|$, $\rhoSz = \ketbra{0}{0}_S$ and $E_S^a = \ketbra{a}{a}_S$ in our implementation of the SDP (\cref{sec:implementation}), there is a direct correspondence between the two scenarios; see \cref{fig:protocoldilation}. Therefore, the upper bounds obtained by the SDP in \cref{eq:finalSDPmain} can be compared with the known achievable values from \cite{vieira2022temporal}, as shown in \cref{tbl:nontrivialSDPresults} and discussed in \cref{sec:numerics}.

This construction applies to deterministic models as they have deterministic transitions between states, i.e., each is described by a single Kraus operator per outcome. We thus have
\begin{proposition}If both $\dE \ge \DC(\va)$ and $\dS \ge n(\va)$, then there is a choice of probe state and measurements such that $\omega^{\va}_{\dE} = 1$.
\end{proposition}
Note that the condition $\dS \ge n(\va)$ is strictly required for deterministic production of $\va$.
In fact, if $p(\va|\dE) = 1$, every symbol must occur deterministically.
Then, at each step the system must be in a state $\rhoS^a$ such that $\Tr{\rhoS^a E_S^a} = 1$. Thus, the measurements $E_S^a$ must be able to perfectly discriminate between the states $\{\rhoS^a\}_a$, which is possible only if $\dS \ge n(\va)$. We have established:
\begin{proposition}If $\dS < n(\va)$, then $\omega^{\va}_{\dE} < 1$ for any choice of probe state or measurements, and any environment dimension $\dE$.
\end{proposition}
In conclusion, these results tells us that nontrivial bounds appear for $\dE < \DC(\va)$, that we can compare these bound with the single-system scenario for  $\dS \ge n(\va)$, and that $|\cA| = 2$ and $L = 3$ is the smallest scenario displaying non-trivial memory effects.

\subsection{Choice of sequence}\label{sec:choiceseq}

As briefly noted in \Cref{sec:detconds}, in order to construct a witness for $\dE > d$, it is vital to choose a sequence $\va$ with sufficient deterministic complexity. If the sequence chosen is too short or too simple (i.e., if $\DC(\va) \le d$), then the bound $\omega^{\va}_d$ is trivial and the dimension $\dE$ cannot be tested with such sequence.

However, due to the difficulty in computing these bounds, and the fact longer sequences generally become less probable---and therefore less likely to violate the witness inequality---generally one should choose sequences to be as short as possible while having deterministic complexity of the same size as the minimum environment dimension one wishes to witness.

As a concrete example, if we wish to witness $\dE > 3$, we should choose a sequence $\va$ with deterministic complexity $\DC(\va) = 4$, e.g., $\va = 0001$, which is the shortest sequence satisfying this requirement. Then, a violation of the bound $\omega^{0001}_3$ informs us that, in fact, $\dE \ge 4$.

\section{Conclusions and outlook}\label{sec:conclusions}

We presented a method to lower bound the dimension of the environment interacting with a probe system, based only on the statistics of measurements performed on the (partially characterized) system, and without any assumption on the environment or the dynamics. This is achieved via a hierarchy of semidefinite programs that upper bound temporal correlations achievable in various experimental scenarios, under the assumption of finite memory. Such bounds can be applied to the detection of the effective environment size in the dynamics of open systems, as well as a certification of the minimum size of an environment's dimension compatible with observations.

To keep the discussion simpler, we applied the optimization for a single sequence. It is straightforward to adapt the objective function to arbitrary linear functions of the full probability distribution $(p(\va|\dE))_{\va}$, as those appearing in the temporal inequalities derived in, e.g., \cite{hoffmann2018,budroni2019memory,spee2020simulating,mao2022strucdimbound}. This may, in principle, lead to better witnesses. We leave the numerical explorations of this problem to a future investigation.

We assumed a joint unitary evolution between system and environment, which leads to a CJ representation of these maps given in terms of symmetric states. As explained in \cref{sec:repunitaries}, this assumption can usually be justified on physical grounds. Nevertheless, a natural question to ask is how do our results change if we consider arbitrary CPTP maps? In that case, the SDP should be modified by replacing the symmetry constraint with permutation invariance, i.e., $V_\sigma \PhiN V^\dagger_\sigma = \PhiN$ for all permutations $\sigma \in \fS_{N}$; see \cref{eq:Vsigma}. We expect that bounds for CPTP maps may be larger than those for unitary channels of the same dimension, especially for objective functions involving more than a single sequence, but numerical optimization is significantly more costly in this case, as the permutation invariant operator basis is of size~\cite[Ch.~7]{watrous2018} $\binomial{N+(\dES)^4-1}{N}$, in contrast to $\binomial{N+(\dES)^2-1}{N}^2$ in the symmetric case. While it is also possible to write such CPTP maps in terms a dilation of the environment, this approach would likely result in more extraneous variables in the SDP, i.e., the terms in the subspace orthogonal to the dilation ancilla's initial state, possibly rendering the optimization intractable.

Additionally, the SDP in \cref{eq:finalSDPmain} and its related implementation techniques are, in fact, quite general, and can be applied to a wide range of scenarios beyond what we have considered here. As the operator $X$ from \cref{eq:Xdef} is fixed, any choice of intermediate operations between each unitary could be chosen. E.g., the initial system and environment states could be correlated, and the intermediate maps $\cM_a$ could each be replaced by arbitrary joint operations, even time-dependent ones. Such approaches could then, for example, be used to bound observables for specific types of processes, e.g., quantum processes with only classical memory~\cite{giarmatzi2021witnessingquantum}, by assuming additional entanglement breaking channels on the environment. Therefore, provided the problem is numerically tractable, our techniques are independent of what explicit measurements are chosen.

While the high dimensionality of the current formulation of the SDP quickly renders general numerical implementations intractable, our approach still offers new avenues for the subject of bounding temporal correlations, and their relationship to open-system dynamics. Ultimately, the techniques developed herein should be taken as a proof-of-concept for future developments and improvements.  It remains to be seen whether more efficient numerical techniques, or even alternative outer approximations, are better suited for addressing such problems. Further investigation of these avenues is subject to future work.

Nevertheless, the success of our approach highlights the wealth of information contained in temporal correlations and the potential of new techniques for characterizing large complex systems by means of a small probe alone, by exploiting non-trivial properties of the temporal correlations achieved by systems of bounded size.

\section{Acknowledgments}
We thank Huan-Yu Ku, Yelena Guryanova and Joshua Morris for valuable discussions and comments. This work is supported by the Austrian Science Fund (FWF) through projects ZK 3 (Zukunftskolleg), Y879-N27 (START project), P 35810-N and P 36633-N (Stand-Alone), F 7113 (BeyondC), by the European Union's Horizon Europe research and innovation programme under the Marie Sk{\l}odowska-Curie grant agreement No. 101068332, and by grant number FQXi-RFP-IPW-1910 from the Foundational Questions Institute and Fetzer Franklin Fund, a donor advised fund of Silicon Valley Community Foundation.

\appendix
\begin{widetext}

\section{Constructing the SDP relaxation}\label{app:theSDP}

This appendix explains in detail the various steps to analytically formulate the maximization problem $\omega_{\dE}^{\va} := \max_U p(\va|\dE)$ as a semidefinite program, with the next appendix (\cref{app:implementation}) focusing on its software implementation. \Cref{fig:relaxationsapp} provides a schematic outline of our approach.

\begin{figure}[H]\centering
	\includegraphics[width=0.75\linewidth]{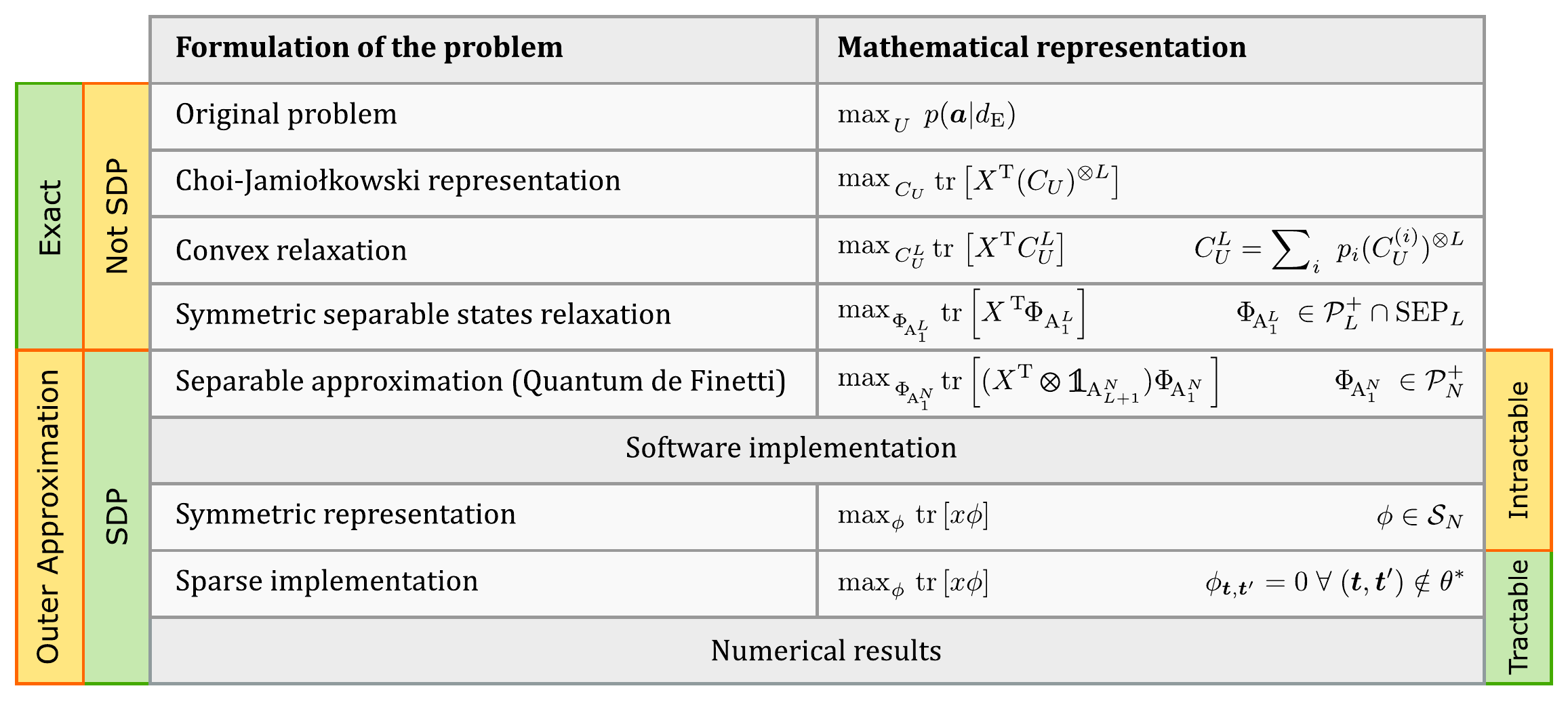}
	\caption{Schematic of all steps undertaken for computing an upper bound for $\omega_{\dE}^{\va}$ by formulating and solving the SDP problem. The tractable/intractable labels, on the right, refer to the case $\dE = 2$.}
	\label{fig:relaxationsapp}
\end{figure}

Let $\ell = 1, \dots, L$ enumerate the time steps, and let $\SI_\ell$ and $\SO_\ell$ be the input and output spaces of the $\ell$-th unitary evolution $\cU$, respectively. For convenience, we write $A_\ell := \SI_\ell \otimes \SO_\ell$ for the joint space at step $\ell$, and $\AIO{a}{b} := A_a \otimes A_{a+1} \otimes \cdots \otimes A_{b-1} \otimes A_b$ for the sequential spaces from $a$ to $b$. We assume that all spaces are isomorphic, i.e., $\SI_\ell \cong \SO_\ell \cong \cHES$ for all $\ell$, but we preserve labels for clarity. We can thus write the respective  maps in the Choi-Jamiołkowski representation (see \cref{eq:cjiso})  as:
\begin{equation}\label{eq:choimatrices}
	C_U := \tfrac{1}{d_{ES}}\cC(\cU), \quad M_a := \cC(\cM_a), \quad \hat{M}_a := \TrP{\cO}{M_a},
\end{equation}
where $\cC(\Lambda)$ is the Choi matrix of the map $\Lambda$ and $C_U$ is normalized such that it corresponds to a quantum state, which is useful during implementation of the symmetric representation of the problem (\cref{app:impsym}). As before, we highlight that the local input ($\cI$) and output ($\cO$) spaces of these maps are interleaved: For the $\ell$-th $\cU$, $\cI = \SI_\ell$ and $\cO = \SO_\ell$, but for $\cM_{a_\ell}$ we have $\cI = \SO_\ell$ and $\cO = \SI_{\ell+1}$. Since the final output state is discarded (i.e., we only are only concerned with the \textit{probability} of outcome sequences), the final $\cO$ is traced out, yielding $\hat{M}_a$ in the equation above. This structure of the spaces is illustrated in \cref{fig:experimentapp}.

\begin{figure}[h!]\centering
	\includegraphics[width=0.5\linewidth]{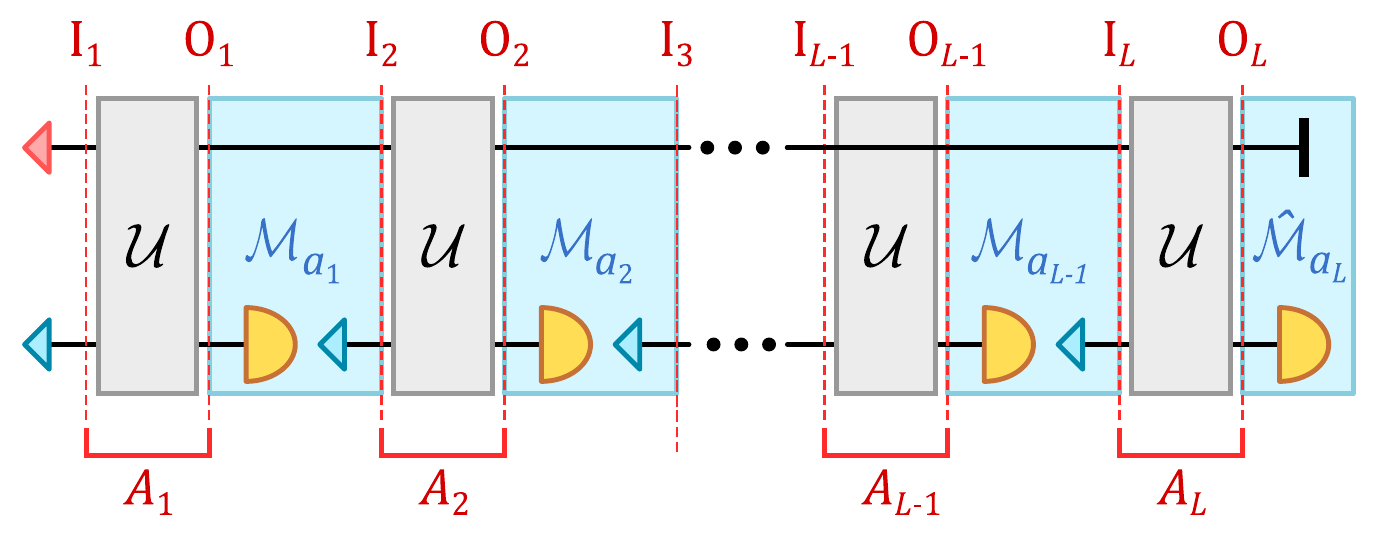}
	\caption{Diagram of the specific protocol discussed in this work, with each time step written in terms of distinct input and output Hilbert spaces, denoted by the dotted lines. Note that the input and output spaces of the $\cU$ and $\cM_{a}$ are interleaved: $\cU$ has inputs $\SI_\ell$ and outputs $\SI_\ell$, but for $\cM_{a_\ell}$, inputs are $\SO_\ell$ and outputs $\SI_{\ell+1}$.}
	\label{fig:experimentapp}
\end{figure}

We may now specify the probability of the sequence $\va$ as
\begin{align}\label{eq:probchoi}
	p(\va|\dE) =& \; \Tr{ X^\T \left( \bigotimes_{\ell=1}^{L} C_U^{\SI_\ell \SO_\ell} \right) } = \Tr{ X^\T (C_U)^{\otimes L} } \\
	X :=& \; (\dES)^L \cdot (\rhoEz \otimes \rhoSz)^{I_1} \otimes  M_{a_1}^{O_1 I_2} \otimes \dots \otimes M_{a_{L-1}}^{O_{L-1} I_L} \otimes \hat{M}_{a_L}^{O_L}.
\end{align}
Where the correcting normalization factor $(\dES)^{L}$ was incorporated into $X$ (to make up for the normalization of $C_U$), and $X^\T$ denotes the transpose with respect to the basis chosen for the isomorphism. In order to obtain an upper-bound for $p(\va|\dE)$, our goal is to optimize \cref{eq:probchoi} over all possible unitaries, in terms of $C_U$. As per the above definitions, it follows that $C_U$ must satisfy the constraints:
\begin{itemize}
	\item $C_U \ge 0$ and $\rank C_U = 1$, as it represents a unitary channel,
	\item $\TrP{\cO}{C_U} = \one_{\cI} / \dES$, as it is a trace preserving map
\end{itemize}

We may thus define our optimization problem as
\optimization{eq:opt1}{The initial formulation of the problem.}{
	\rhoEz,\; \rhoSz,\; \{ M_a \}_a,\; \va
}{
	\max_{C_U} \; \Tr{X^\T (C_U)^{\otimes L}}
}{
	C_U \ge 0, \quad \rank C_U = 1, \quad \TrP{\cO}{C_U} = \one_{\cI} / \dES.
}
This may be immediately relaxed to a convex form, without affecting the maximum of the objective function, as
\optimization{eq:opt2}{The convex relaxation of the original problem.}{
	\rhoEz,\; \rhoSz,\; \{ M_a \}_a,\; \va
}{
	\max_{\{\ket{\phi_i}\}_i} \Tr{X^\T \PhitL},\qquad {\PhitL} = \sum_i p_i \purestate{\phi_i}^{\otimes L},\; p_i \ge 0,\; \sum_i p_i = 1
}{
	\TrP{\cO}{\purestate{\phi_i}} = \one_{\cI} / \dES, \quad \ket{\phi_i} \in \cHES \otimes \cHES.
}

However, this optimization problem is non-linear, not only due to the tensor product in $\purestate{\phi_i}^{\otimes L}$, but also on the rank-1 nature of $\purestate{\phi_i}$. Let us now define our \emph{target set} $\cT$ as
\begin{equation}\label{eq:targetsetdef}
	{\PhitL} \in \cT := \set{ \sum_i p_i \purestate{\phi_i}^{\otimes L} \cond \ket{\phi_i} \in \cHES \otimes \cHES,\; p_i \ge 0,\; \sum_i p_i = 1,\; \TrP{\cO}{\purestate{\phi_i}} = \one_{\cI} / \dES }.
\end{equation}
Elements of this set will not generally be rank-1, but the optimal solutions of Problems 1 and 2 belong to this set. Our goal now is to approach $\cT$ by means of further relaxations, which is achieved by exploiting the symmetry of ${\PhitL}$.

\subsection{The unitary channel constraints}

The symmetric structure of ${\PhitL}$ requires all subspaces to be in the same local state. We can relax this by considering instead the separable set on $L$ parties,
\begin{equation}\label{eq:PhiL}
	\PhiL \in \SEP_L := \conv \set{ \bigotimes_{\ell=1}^{L} Q_\ell \cond Q_\ell \in \cB(\cHES \otimes \cHES),\; Q_\ell \ge 0,\; \Tr{Q_\ell} = 1 },
\end{equation}
such that all local spaces become independent. Here, $\conv$ denotes the convex hull of the set, and $\cB$ the set of bounded operators. We must now find ways to restore rank-1 and permutation invariance constraints for the optima over this set. For now, we highlight that $\cT \subset \SEP_L$, but simply switching from ${\PhitL} \in \cT$ to $\PhiL \in \SEP_L$ \emph{does} change our problem, so that we must impose further constraints to the set $\SEP_L$ to restore the original problem in $\cT$. The first step is restoring symmetry and optimality of rank-1 to $\PhiL$. After this, we may restore the partial trace constraint.

Let $\sigma \in \fS_n$ be a permutation from the set of all permutations on $n$ symbols $\fS_n$. We define
\begin{align}
	V_\sigma &:= \sum_{i_1, \cdots, i_n = 1}^{d}  \vert i_{\sigma^{-1}(1)}, \cdots, i_{\sigma^{-1}(n)} \rangle \hspace{-0.1em} \langle i_1, \cdots, i_n \vert, \label{eq:Vsigma} \\
	{P^{+}_n} &:= \frac{1}{n!} \sum_{\sigma \in \fS_n} V_\sigma \label{eq:Pplus}
\end{align}
as, respectively, the permutation operator and projector onto the symmetric subspace. The symmetric subspace on $n$ copies of $\cH$ is defined as $\Symn{\cH}:= \set{ \ket{\psi}\in \cH^{\otimes n} \cond {P^{+}_n} \ket{\psi}=\ket{\psi} }$, and the space of symmetric operators is given by the set $\Symn{\cB(\cH)} := \{ A \in \cB(\cH^{\otimes n})\ |\ P^+_n A = A\}$; see \cite{harrow2013church}. Note that, as $\rho$ is Hermitian, the symmetry condition can be equivalently written either as ${P^{+}_n} \rho = \rho$ or ${P^{+}_n} \rho {P^{+}_n} = \rho$. In fact, if ${P^{+}_n} \rho = \rho$, then $\rho^\dagger = ({P^{+}_n} \rho)^\dagger = \rho^\dagger P^{+} = \rho P^{+} = \rho \Rightarrow {P^{+}_n} \rho {P^{+}_n} = \rho$. Conversely, $\rho = {P^{+}_n} \rho {P^{+}_n} = ({P^{+}_n})^2 \rho {P^{+}_n} = {P^{+}_n} \rho$. More details about the symmetric subspace can be found in \cref{app:impsym}.

\subsubsection{Ensuring symmetry and rank-1 at the optimum}
\label{subsec::SymRank1}

By utilizing $\PhiL$ as our variable as in \cref{eq:PhiL}, we have lost both symmetry and rank-1 guarantees at the optimum. However, both can be restored by ensuring $\PhiL$ is in the symmetric subset of $\SEP_L$; see, e.g., \cite{Korbicz2005}.  For completeness, we provide in the following an elementary proof that
\begin{equation}\label{eq:symeqpi}
\Symn{\cB(\cH)} \intersection \SEP_n = \set{ \sum_i p_i \purestate{\phi_i}^{\otimes n} \cond \ket{\phi_i} \in \cHES \otimes \cHES,\quad \braket{\phi_i}{\phi_i} = 1, \quad p_i \ge 0,\quad \sum_i p_i = 1 },
\end{equation}
which implies $\cT \subset \operatorname{Sym}_L(\cB(\cH)) \intersection \SEP_L$. In fact, consider $\rho \in \Symn{\cB(\cH)} \intersection \SEP_n$. Since it is separable, we can write it as as a convex mixture of pure product states, i.e.,
\begin{equation}\label{eq:rhosymdec}
	\rho = \sum_i p_i \purestate{\Phi_i},
\end{equation}
with $\ket{\Phi_i}$ a $n$-party pure product state, i.e., $\ket{\Phi_i} = \bigotimes_{j=1}^{n} \ket{\phi_{i,j}}$, with the $\{\ket{\phi_{i,j}}\}_j$ not necessarily equal for a given $i$.
However, $\rho \in \Symn{\cB(\cH)}$ implies $\range{\rho} \subseteq \Symn{\cH}$, and if we prove that $\range{\purestate{\Phi_i}} \subseteq \Symn{\cH}$ for all $i$, then we find that $\{\ket{\phi_{i,j}}\}_j$ must be identical for each $i$, otherwise $\ket{\Phi_i}$ fails to be symmetric, namely, $\purestate{\Phi_i} = \purestate{\phi_i}^{\otimes n}$ and, thus, $\rho = \sum_i p_i \purestate{\phi_i}^{\otimes n}$.
It remains to prove that $\range{\purestate{\Phi_i}} \subseteq \Symn{\cH}$. To show this, we note that by \cref{eq:rhosymdec} and positivity of $\rho$ and $\purestate{\Phi_i}$, we have $\kernel{\rho} = \intersection_i \kernel{\purestate{\Phi_i}}$. By Hermiticity, it follows that $\range{\rho} = (\kernel{\rho})^{\perp} = \text{span}(\union_i \range{\purestate{\Phi_i}})$, which concludes the proof.

Importantly, \cref{eq:symeqpi}, restores the symmetry and rank-1 properties for optimum solutions of \cref{eq:opt2} when $\PhiL$ is restricted to $\Symn{\cB(\cH)} \intersection \SEP_n$. In light of this, in the following we use $\SymSep_n := \Symn{\cB(\cH)} \intersection \SEP_n$.

\subsubsection{The trace-preserving constraint}\label{sec:TPcond}

We now restore the trace-preserving constraint $\TrP{\cO}{\purestate{\phi_i}} = \one_{\cI} / \dES$ for states expressed as in \cref{eq:symeqpi}. This is established by proving that it is sufficient to satisfy this condition for a single party in the convex mixture.

\begin{proposition}\label{prop:sepchoicond}For any\; $\PhiL \in \SymSep_L$:
	\begin{equation}
		\TrP{\cO_1}{\PhiL} = \frac{\one_{\cI_1}}{\dES} \otimes \TrP{\SA{1}}{\PhiL} \quad\iff\quad \TrP{\cO_\ell}{\purestate{\phi_i}_{\SA{\ell}}} = \frac{\one_{\cI_\ell}}{\dES} \otimes \purestate{\phi_i}^{\otimes L-1} \quad \forall \ell.
	\end{equation}
\end{proposition}

\begin{proof}
	We provide a generalization of the arguments in \cite[App.~A]{yu2020quantuminspired}. From \cref{eq:symeqpi}, if $\PpL \PhiL = \PhiL$, we know it admits a decomposition of the form
	\begin{equation}\label{eq:phiabconvex}
		\PhiL = \sum_i p_i \purestate{\phi_i}^{\otimes L}, \qquad p_i \ge 0,\quad \sum_i p_i = 1.
	\end{equation}
	We introduce the one-party auxiliary map $\cE(\cdot) = \TrP{\cO}{\cdot} - \Tr{\cdot} \one_{\cI} / \dES$ implementing our partial trace constraint in $\cT$, such that $( \cE_{\SA{1}} \otimes \id_{\SA{2}} \otimes \id_{\AIO{3}{L}} ) ( \PhiL ) = 0$
	for any $\PhiL \in \SymSep_L \intersection \, \cT$, and, due to symmetry, the same is true if the map had been applied to the second party instead. Defining $\tilde{\cE}(\cdot) = [\cE(\cdot)]^\dagger$, we may then write $( \cE_{\SA{1}} \otimes \tilde{\cE}_{\SA{2}} \otimes \id_{\AIO{3}{L}} ) ( \PhiL ) = 0$. As $\cE$ acts on Hermitian operators and is Hermiticity-preserving, we have $\cE = \tilde{\cE}$ and the explicit distinction is only made for added clarity in the following proof. Applying the map to the convex mixture in \cref{eq:phiabconvex}, we obtain
	\begin{equation}
		\sum_i p_i \cE_{\SA{1}}(\purestate{\phi_i}_{\SA{1}}) \otimes \tilde{\cE}_{\SA{1}}(\purestate{\phi_i}_{\SA{2}}) \otimes \purestate{\phi_i}^{\otimes L-2} = \sum_i p_i E_i \otimes E_i^\dagger \otimes \purestate{\phi_i}^{\otimes L-2} = 0,
	\end{equation}
	where $E_i = \cE(\purestate{\phi_i})$. We need this to be true for all $E_i$ if we want to ensure $\PhiL$ obeys our partial trace constraint globally in the convex mixture. To prove this is the case, let $G \in \Complexes^{\dES^2 \times \dES^2}$, so that
	\begin{equation}
		\Tr{(G \otimes G^\dagger \otimes \one_{\AIO{3}{L}}) \left(\sum_i p_i E_i \otimes E_i^\dagger \otimes \purestate{\phi_i}^{\otimes L-2} \right)} = \sum_i p_i |\hspace{-0.15em}\Tr{G E_i}\hspace{-0.15em}|^2 \Tr{ \purestate{\phi_i}^{\otimes L-2} } = 0,
	\end{equation}
	which must hold for any $G$. Since $p_i \ge 0$, $|\hspace{-0.15em}\Tr{G E_i}\hspace{-0.15em}| \ge 0$, and $\Tr{\purestate{\phi_i}^{\otimes L-2}} = 1$, this can only be true for all $G$ if all $E_i = 0$. 	Furthermore, by the symmetry of each term in the mixture, it is then sufficient to ensure the constraint is satisfied in a single party. From \cref{eq:symeqpi}, the converse statement follows trivially.
\end{proof}

Thus, any state $\PhiL \in \SymSep_L$ satisfying $\TrP{\SO_1}{\PhiL} = (\one_{\SI_1} / \dES) \otimes \TrP{\SA{1}}{\PhiL}$ is part of the target set $\cT$ and vice versa.

\subsubsection{The final rank-constrained problem over SEP}

Given all of the previous results, we have now established a relation between the sets $\cT$ and $\SEP_L$,
\begin{equation}
	\cT = \set{ \PhiL \in \SEP_{L} \cond \PpL \PhiL = \PhiL,\quad \TrP{\SO_1}{\PhiL} = (\one_{\SI_1} / \dES) \otimes \TrP{\SA{1}}{\PhiL} },
\end{equation}
which allows us to remove the non-linear dependence on the tensor product $(C_U)^{\otimes L}$, by replacing the search space with $\SymSep_L$, while obeying the linear constraints on the partial trace of one party. We thus obtain:
\optimization{eq:opt3}{Rank-1 and symmetric optimum (exact)}{
	\rhoEz,\; \rhoSz,\; \{ M_a \}_a,\; \va
}{
	\max_{\PhiL \, \in \; \SEP_L} \Tr{X^\T \PhiL}
}{
	\PhiL \ge 0, \quad \Tr{\PhiL} = 1, \quad \PpL \PhiL = \PhiL,\\
	&	\TrP{\SO_1}{\PhiL} = \frac{\one_{\SI_1}}{\dES} \otimes \TrP{\SA{1}}{\PhiL}
}

However, $\PhiL \in \SEP_L$ is still not a linear constraint and highly non-trivial to be enforced, as a full characterization of the separable set is NP-HARD~\cite{gurvits2003}. Therefore, while a priori appearing more manageable, the above problem is still not in a form that can be tackled numerically. Fortunately, the symmetric constraint can also be exploited to obtain approximations of the symmetric separable states.

\subsection{The SEP constraint}\label{app:SDPdeFinetti}

\subsubsection{Approximation of SEP via the quantum de Finetti theorem}

Ensuring $\PhiL \in \SEP_L$ cannot be done exactly, so as it stands the problem is still numerically intractable. However, it is possible to define arbitrarily-precise outer approximations of the separable set through symmetric extensions and the quantum de Finetti theorem~\cite{Caves2002,christandl2007one}. Concretely, we can approximate, in principle to any desirable precision, a $\PhiL \in \SEP_L$ by a marginal over a larger symmetric (and not necessarily separable) state $\PhiN$, provided $N$ is sufficiently large. For any finite $N$, this leads to an outer approximation of $\cT$, which establishes a convergent hierarchy of outer approximations, such that for any $N \ge L$, we have
\begin{align}\label{eq:apphierarchybounds}
	\begin{split}
		&\omega^{\va}_{\dE} \le \cdots \le \tilde{\omega}^{\va,N+1}_{\dE} \le \tilde{\omega}^{\va,N}_{\dE} \le \cdots \le \tilde{\omega}^{\va,L}_{\dE}, \\
		&\text{with}\quad \lim_{N \to \infty} \tilde{\omega}^{\va,N}_{\dE} = \omega^{\va}_{\dE}.
	\end{split}
\end{align}
More precisely, using the quantum de Finetti theorem we can establish asymptotic error bounds on this approximation~\cite[Cor.1]{chiribella2011}. In terms of the trace norm, this asymptotic error bound is given by
\begin{equation}
	\left\Vert \PhiL - \TrP{\AIO{L+1}{N}}{\PhiN} \right\Vert_1 \le \frac{2L(L+(\dES)^2+1)}{N + (\dES)^2}.
\end{equation}
As the symmetry requirement applies to both $\PhiL$ and $\PhiN$, this application of the quantum de Finetti theorem for approximating $\PhiL \in \SEP_L$ is straightforward in the SDP in \cref{eq:finalSDPmain}, only requiring the use of a larger state $\PhiN$ as the optimization variable, and a suitable adjustment of the objective function using $(X^\T \otimes \one_{\AIO{L+1}{N}})$.

\subsubsection{Additional separability constraints}\label{app:ppt}

To improve the approximation, at an extra computational cost, we may also include additional separability constraints. In our case, we consider the constraints of positive partial transpose (PPT) on $\PhiN$:
\begin{equation}\label{eq:PPTconst}
	\PhiN^\Ta \ge 0, \quad \forall \; \text{non-equivalent bipartitions} \; \alpha.
\end{equation}
Here, ``non-equivalent bipartitions'' refers to the fact that it is unnecessary to include all bipartitions, as the state $\PhiN$ is symmetric, and thus permutation invariant. Therefore, the bipartitions needed are $\alpha_k$, where, e.g., we transpose only the first $k$ subspaces, for $1 \le k \le \floor{N/2}$, with symmetry taking care of the remaining constraints.

The existence of entangled states which have positive partial transpose~\cite{horodeckis1998pptes} means these constraints reduce the feasible convex set to that of PPT states, not the separable states. While entangled PPT states are not suitable solutions to the original problem, they are still adequate approximations and provide upper bounds for the actual optimal values. Inclusion of PPT constraints will perform at least as well as optimizing without them (i.e., the solution can only be improved), but convergence is significantly improved in certain cases~\cite{navascues2009} to $O(1/N^2)$ as opposed to $O(1/N)$ in the absence of these constraints. As shown in \cref{tbl:trivialSDPresults}, for the $\dE = 1$ case PPT constraints were sufficient to provide the exact analytical bounds.

\subsubsection{SDP for {\SEP} relaxation}

By all of the above results, we arrive at the following SDP relaxation of the original problem:
\optimization{eq:opt4}{Final outer approximation (SDP).}{
	\displaystyle \dE,\; \dS,\; \rhoEz,\; \rhoSz,\; \{ M_a \}_a,\; \va
}{
	\tilde{\omega}_{\dE}^{\va,N} := \max_{\PhiN} \Tr{(X^\T \otimes \one_{\AIO{L+1}{N}}) \PhiN}
}{
	\PhiN \ge 0, \quad \Tr{\PhiN} = 1, \quad \PpN \PhiN = \PhiN,\\
	&\TrP{\SO_1}{\PhiN} = \frac{\one_{\SI_1}}{\dES} \otimes \TrP{\SA{1}}{\PhiN} \\
	&\PhiN^{\Ta} \ge 0, \forall \alpha \in \mathfrak{A},
}
with $\mathfrak{A}$ denoting the set of non-equivalent bipartitions (\cref{app:ppt}). The solutions of this SDP provide upper-bounds for the maxima $\omega_{\dE}^{\va}$.

\section{Implementation}\label{app:implementation}

While in principle numerically accessible, the SDP in \cref{eq:opt4} still contains too many variables to allow for the computation of upper bounds of $\tilde{\omega}_{\dE}^{\va,N}$, even for low-dimensional cases. This problem can be tackled by exploiting the fact that all appearing objects can be expressed within the symmetric subspace, as well as the (potential) sparsity of the problem.

In our case, a further simplification is possible by first noticing that both $X$ and $\PpN$ are real-valued. If we additionally consider a real-valued basis for the partial trace, we conclude that, for any feasible $\PhiN$, its entry-wise complex conjugate $\PhiN^*$ is also feasible, while providing the same value for the objective function. Therefore, we may perform the optimization using a real-valued $\PhiN$.
Next, we explain how to exploit the symmetry property.

\subsection{Symmetric representation}\label{app:impsym}

The symmetry constraint can be satisfied automatically by expressing $\PhiN$ directly in terms of a basis for the symmetric subspace. This allows for a great reduction in the number of variables and linear constraints in the SDP, which is required for its efficient numerical optimization. In the following, we describe how to construct this basis for the symmetric subspace, largely based on \cite{harrow2013church}, and how to adapt the remaining constraints to act directly in this symmetric representation, based on a generalization of results provided in \cite{stockton2003} for the qubit case.

We begin by defining a canonical basis for the symmetric subspace. With the convention $\Naturals=\{1,2,3,\ldots\}$ and $\Naturals_0 := \Naturals \union \{0\}$, for $d,n \in \Naturals$, let
\begin{equation}
	{\fT_d^n} := \set{ (t_1, t_2, \dots, t_d) \cond t_i \in \Naturals_0,\quad \sum_{i=1}^{d} t_i = n }
\end{equation}
be the set of weak integer compositions for $n$ into exactly $d$ parts, possibly of zero size, which we refer to as \emph{types}. These types form a canonical labeling for the basis of the symmetric subspace.
Now, defining $[d] := \{1, 2, \dots, d\}$, let $\vu \in [d]^n$, i.e., $\vu = (u_1, \dots, u_n)$ with $u_\ell \in [d]$, and $T(\vu) = \vt$ be the type of the vector $\vu$, where $t_i$ counts the number of instances where $u_\ell = i$ holds.  As a concrete example, if $d = 6$ and $n = 8$, we might have:
\begin{equation}
	\vu = (1, 4, 1, 2, 3, 2, 2, 6) \qquad \longrightarrow \qquad T(\vu) = (2,3,1,1,0,1).
\end{equation}|
In words, types count how many times a number $1, \dots, d$ occurs in $\vu$, and therefore are invariant under permutations, i.e., $T(\vu) = T(P_\sigma \vu)$ for any permutation operator $P_\sigma$ acting on entries of $\vu$. Given a Hilbert space $\cH$ with $d = \dim \cH$, we may now define a basis for $\Symn{\cH}$ in terms of the types $\vt \in {\fT_d^n}$ by constructing the non-normalized and normalized orthogonal basis vectors for the symmetric subspace~\cite{harrow2013church}, respectively, as
\begin{align}\label{eq:symbasis}\begin{split}
	\ket{\text{Sym}(\vt)} := \sum_{\vu ; T(\vu) = \vt} \ket{u_1, u_2, \cdots, u_n}, \qquad
	\ket{\text{sym}(\vt)} := (\vt)^{-1/2} \ket{\text{Sym}(\vt)},
\end{split}\end{align}
where $(\vt) = \frac{(\sum_i t_i)!}{t_1! \cdot t_2! \cdots t_d!}$ denotes the multinomial coefficient for the normalization, and, for convenience in notation, we adopt the canonical basis $\{\ket{u}\}_{u=1}^d$ for each individual $u_i$. \Cref{eq:symbasis} clearly defines a symmetric state, since each $\vu$ occurs only once in the sum. Thus, for any permutation operator $V_\sigma$, we have $V_\sigma \ket{\text{sym}(\vt)} = \ket{\text{sym}(\vt)}$, and similarly for the non-normalized $\ket{\text{Sym}(\vt)}$.

From \cref{eq:symeqpi}, we know any symmetric separable state can be written as $\sum_i p_i \purestate{\phi_i}^{\otimes n}$. Therefore, by expressing $\ket{\phi_i}^{\otimes n}$ in terms of the symmetric basis in \cref{eq:symbasis}, we may write any symmetric operator by indexing the degrees of freedom of $\PhiN$ by the types $\vt, \vt'$. In our current problem, with a local dimension $d = (\dE \dS)^2 = \dES^2$, the symmetric space has total dimension~\cite{harrow2013church}
\begin{equation}
	D_\cS = \dim \operatorname{Sym}_N(\cH) = \multiset{\dES^2}{N} = \binom{N + \dES^2 - 1}{N},
\end{equation}
with $\multiset{d}{n}$ the multiset notation, i.e., the number of ways of picking $n$ elements out of $d$, with repetitions allowed, which corresponds to the cardinality of ${\fT_d^n}$. This is a significant reduction from the original $\dim \AIO{1}{N} = (\dES)^{2 N}$.

For simplicity in notation, we consider the dependence of $D_\cS$ on $\dES$ and $N$ implicit in the following. To concretely exploit this reduction of the number of variables, i.e., in order to write $\PhiN$ in terms of a smaller matrix, we re-express the problem as follows. Let $\cS_N := \Complexes^{D_\cS \otimes D_\cS}$, and $\phi \in \cS_N$. We specify elements $\phi_{\vt,\vt'}$, with $\vt,\vt' \in \fT^N_{\dES^2}$, by a canonical orthonormal basis $\{\ket{\vt}_{\cS_N} \}_{\vt}$. We define an isometry between the two representations of the symmetric subspace, i.e., the normalized symmetric vectors in the full space $\AIO{1}{N}$ and the canonical basis for $\cS_N$, as follows:
\begin{equation}
	V_N = \sum_{\vt} \ket{\vt}_{\cS_N} \bra{\text{sym}(\vt)}_{\AIO{1}{N}}.
\end{equation}
Note that this isometry acts as a projector in $\AIO{1}{N}$, i.e., $V_N {V_N}^\dagger = \one_{\cS_N}$ but ${V_N}^\dagger V_N  \ne \one_{\AIO{1}{N}}$. Therefore, if $\PpN \PhiN = \PhiN$, we may write it concisely as:
\begin{equation}\label{eq:appsymversion}
	\begin{split}
		\PhiN = {V_N}^\dagger \phi {V_N} = \sum_{\vt,\vt'} \phi_{\vt,\vt'} \ketbra{\text{sym}(\vt)}{\text{sym}(\vt')}_{\AIO{1}{N}}, \\
		\phi = V_N \PhiN {V_N}^\dagger, \quad \text{and} \quad x = {V_N} (X^\T \otimes \one_{\AIO{L+1}{N}}) {V_N}^\dagger.
	\end{split}
\end{equation}
so that the objective function, which acts entirely within the symmetric subspace, can be written as
\begin{equation}
	\Tr{(X^\T \otimes \one_{\AIO{L+1}{N}}) \PhiN} = 
	\Tr{ (X^\T \otimes \one_{\AIO{L+1}{N}}) {V_N}^\dagger \phi {V_N}} = \Tr{ {V_N}(X^\T \otimes \one_{\AIO{L+1}{N}}) {V_N}^\dagger \phi } =\Tr{x \phi}.
\end{equation}

Positivity of $\PhiN$ can be guaranteed by requiring $\phi \ge 0$ directly, as $\phi$ and $\PhiN$ have the same non-zero eigenvalues. With the chosen normalization of $C_U$, we may also write $\Tr{\phi} = 1$. In practice, it is sufficient to numerically compute $x$ directly by means of \cref{eq:appsymversion}, as computing the projection of $X^\T \otimes \one_{\AIO{L+1}{N}}$ onto the symmetric subspace analytically can be cumbersome, and this computation only needs to be performed once before the numerical optimization.

\medskip

To rewrite the partial trace constraints of the SDP directly in this symmetric basis, we must find a way to express partial traces in terms of the symmetric representation. In the following, we generalize some results known for qubits~\cite{stockton2003} to arbitrary local dimension. But first, recall that the constraint we wish to rewrite is given by
\begin{equation}\label{eq:ptrconsts}
	\TrP{\SO_1}{\PhiN} = \frac{\one_{\SI_1}}{\dES} \otimes \TrP{\SA{1}}{\PhiN}.
\end{equation}

Treating types as ordinary vectors, we may define addition between types with the same $d$, such that if $\vr \in {\fT_d^n}$ and $\vs \in {\fT_d^m}$, then $\vr + \vs = \vt \in {\fT_d^{n+m}}$. This corresponds to the fact that ${\cS_{n+m}} \subset {\cS_n} \otimes {\cS_m}$~\cite{stockton2003}. Generalizing upon this idea, we may split an $n$-partite symmetric state into $m$ parts of sizes $k_i$, with $\sum_{i=1}^m k_i = n$, through the decomposition
\begin{equation}\label{eq:symsplit}
	\ket{\text{Sym}(\vt)}_{\AIO{1}{N}} = \sum_{\vr_1 + \cdots + \vr_m = \vt} \ket{\text{Sym}(\vr_1)}_{\text{B}_{1}} \cdots \ket{\text{Sym}(\vr_m)}_{\text{B}_{m}},
\end{equation}
where the sum is over all tuples $(\vr_1, \dots, \vr_m) \in ({\fT_d^{k_1}} \times \cdots \times {\fT_d^{k_m}})$ satisfying $\vr_1 + \cdots + \vr_m = \vt$. Here, $\text{B}_{\ell} := \AIO{b_\ell+1}{b_\ell+k_\ell}$, with $b_\ell = \sum_{i=1}^{\ell-1} k_i$, corresponding to the subspace of the $\ell$-th part of the decomposition. Importantly, this decomposition is performed in the full space with non-normalized vectors. For completion, the normalized version of the decomposition in \cref{eq:symsplit} is given by
\begin{equation}\label{eq:symsplitnorm}
	\ket{\text{sym}(\vt)}_{\AIO{1}{N}} = \sum_{\vr_1 + \cdots + \vr_m = \vt} \left(\frac{(\vt)}{(\vr_1)\cdots(\vr_m)}\right)^{-1/2} \ket{\text{sym}(\vr_1)}_{\text{B}_{1}} \cdots \ket{\text{sym}(\vr_m)}_{\text{B}_{m}}.
\end{equation}
However, working under such normalization is cumbersome and inefficient due to the several coefficients involved. Instead, we have chosen to normalize directly in terms of the original types $\vt, \vt'$, which makes normalization straightforward. To adapt the partial trace constraint, we first write the state $\PhiN$ as in \cref{eq:symversion}, and using \cref{eq:symsplit} split the basis into two parts of sizes $(1,N-1)$, obtaining
\begin{align}\begin{split}\label{eq:phinsplit}
		\PhiN = \sum_{\vt,\vt'} \; \sum_{\vr + \vs = \vt} \; \sum_{\vr' + \vs' = \vt'} \hat{\phi}_{\vt,\vt'} \ketbra{\text{Sym}(\vr)}{\text{Sym}(\vr')}_{\SA{1}} \otimes \ketbra{\text{Sym}(\vs)}{\text{Sym}(\vs')}_{\AIO{2}{N}},
\end{split}\end{align}
where the proper normalization is now taken care of by defining $\hat{\phi}_{\vt,\vt'} := \phi_{\vt,\vt'} \left((\vt)(\vt')\right)^{-1/2}$. Since $\vr$ and $\vr'$ are types for a single party, we have that $\ket{\text{Sym}(\vr)}_{\SA{1}}$ are simply vectors in the canonical basis $\{\ket{u}\}_{u=1}^d$. Thus, we will adopt the notation $\ket{u(\vr)}_{\SA{1}} = \ket{\text{Sym}(\vr)}_{\SA{1}}$ in what follows. The tensor product form allows for the simplified application of the partial traces in \cref{eq:ptrconsts}. By linearity of the partial trace, we can treat each term of \cref{eq:phinsplit} separately, so that the terms on the left ($Y_L$) and right ($Y_R$) hand sides of \cref{eq:ptrconsts} can be written as
\begin{align}\label{eq:trpsymsL}
	\begin{split}
		Y_L :=& \TrP{\SO_1}{\ketbra{u(\vr)}{u(\vr')}_{\SA{1}} \otimes \ketbra{\text{Sym}(\vs)}{\text{Sym}(\vs')}_{\AIO{2}{N}}} \\ =& \TrP{\SO_1}{\ketbra{u(\vr)}{u(\vr')}_{\SA{1}}} \otimes \ketbra{\text{Sym}(\vs)}{\text{Sym}(\vs')}_{\AIO{2}{N}}
	\end{split}
\end{align}
\begin{align}\label{eq:trpsymsR}
	\begin{split}
		Y_R :=& \frac{\one_{\SI_1}}{\dES} \otimes \TrP{\SA{1}}{\ketbra{u(\vr)}{u(\vr')}_{\SA{1}} \otimes \ketbra{\text{Sym}(\vs)}{\text{Sym}(\vs')}_{\AIO{2}{N}}} \\ =& \Tr{\ketbra{u(\vr)}{u(\vr')}_{\SA{1}}} \cdot \frac{\one_{\SI_1}}{\dES} \otimes \ketbra{\text{Sym}(\vs)}{\text{Sym}(\vs')}_{\AIO{2}{N}}.
	\end{split}
\end{align}
A few simplifications are now evident. First, we observe that $\Tr{\ketbra{u(\vr)}{u(\vr')}_{\SA{1}}} = \delta_{\vr,\vr'}$, where $\delta_{\vr,\vr'}$ is the Kronecker delta. Secondly, we may split the states $\ket{u(\vr)}_{\SA{1}}$ into the local input and output spaces
\begin{equation}
	\ket{u(\vr)}_{\SA{1}} = \ket{i(\vr)}_{\SI_1} \ket{o(\vr)}_{\SO_1},
\end{equation}
so that we obtain
\begin{align}\begin{split}
	\TrP{\SO_1}{\ketbra{u(\vr)}{u(\vr')}_{\SA{1}}} &= \TrP{\SO_1}{ \ketbra{i(\vr)}{i(\vr')}_{\SI_1} \otimes \ketbra{o(\vr)}{o(\vr')}_{\SO_1} } \\
	&= \ketbra{i(\vr)}{i(\vr')}_{\SI_1} \cdot \Tr{ \ketbra{o(\vr)}{o(\vr')}_{\SO_1} } \\
	&= \ketbra{i(\vr)}{i(\vr')}_{\SI_1} \cdot \delta_{o(\vr),o(\vr')}.
\end{split}\end{align}

Using the above results, and correcting for the missing normalizations, we can apply the isometry $\one_{\SI_1} \otimes V_{N-1}$, so that \cref{eq:trpsymsL,eq:trpsymsR} become:
\begin{align}\label{eq:trpsymsdeltas}\begin{split}
	Y_L &\to ((\vs)(\vs'))^{1/2} \cdot \delta_{o(\vr),o(\vr')} \cdot \ketbra{i(\vr)}{i(\vr')}_{\SI_1} \otimes \ketbra{\vs}{\vs'}_{\cS_{N-1}} \\
	Y_R &\to ((\vs)(\vs'))^{1/2} \cdot \frac{\delta_{\vr,\vr'}}{\dES} \cdot \sum_{i=1}^{\dES} \ketbra{i}{i}_{\SI_1} \otimes \ketbra{\vs}{\vs'}_{\cS_{N-1}}.
\end{split}\end{align}
Here, we have written $\one_{\SI_1} = \sum_{i=1}^{\dES} \ketbra{i}{i}_{\SI_1}$ as to make the $\SI_1 \otimes \cS_{N-1}$ decomposition explicit in both expressions. By inserting \cref{eq:trpsymsdeltas} back into the sum of \cref{eq:phinsplit}, we can appreciate the fact that \cref{eq:trpsymsdeltas} neatly separates each term of the sum as square matrices of size $\dES \times \multiset{\dES^2}{N-1}$, written in terms of $\ketbra{i}{i'}_{\SI_1} \otimes \ketbra{\vs}{\vs'}_{\cS_{N-1}}$. The constraint of \cref{eq:ptrconsts} then tells us we must sum over all $\vt,\vt'$ on both sides, where we can then apply the equality constraint element-wise by simply matching the resulting $(i,\vs,i',\vs')$ entries. As the extra normalization $((\vs)(\vs'))^{1/2}$ of \cref{eq:trpsymsdeltas} is always equal between these elements, it cancels out in their element-wise equality constraint. Thus, it is sufficient to use the normalization of $\hat{\phi}_{\vt,\vt'}$.

During implementation in software, the summations in \cref{eq:phinsplit} can be efficiently performed by taking into account the fact that valid decompositions of the form $\vt = \vr + \vs$ are very restricted in number, and can be efficiently enumerated, grouped, filtered and counted. By assigning a tuple $(i,o,\vs)$ to each $\vt$, the tracing over inputs and output spaces, and applications of the Kronecker deltas, can thus be efficiently computed. Therefore, the linear constraints between variables $\phi_{\vt,\vt'}$ can be constructed by bucketing terms $\phi_{\vt,\vt'}$ by matching $\ketbra{i}{i'}_{\SI_1} \otimes \ketbra{\vs}{\vs'}_{\cS_{N-1}}$, and normalization is simplified by attaching it to $\phi_{\vt,\vt'}$ at the very end, i.e., using $\hat{\phi}_{\vt,\vt'}$.

A similar strategy as above may be employed to define partial transpositions through the symmetric representation, as was done in \cite{stockton2003} for the qubit case, in order to implement the PPT condition. However, for $\dE = 2$ in our case, this approach would render the matrices and the number of constraints too large for a viable computation, once again, and therefore we did not pursue a generalization of this idea. The optimizations we have performed for $\dE = 1$ involving PPT constraints, as shown in \cref{tbl:trivialSDPresults}, used $\PhiN$ directly following \cref{eq:symversion}, which was still tractable in this case. 

While the above steps lead to a significantly smaller representation of the problem, it is not generally a sufficient reduction in variables and constraints for the SDP to be numerically tractable. For concreteness, in the smallest non-trivial case of $\dE = \dS = 2$, this gives for $N = 2, \dots, 5$ a symmetric space of size $D_\cS = 136,\; 816,\; 3\,876,\; 15\,504, \; \dots,$ such that the SDP would be written in terms of $(D_\cS)^2 = 18\,496,\; 665\,856,\; 15\,023\,376,\; 240\,374\,016, \dots$ variables. These examples illustrate how quickly our problem can become numerically intractable, even with symmetry taken into account.

\subsection{Sparse implementation of the SDP}\label{app:impsparse}

In the following, we describe in detail the algorithm we have developed to obtain a sparse version of the SDP. This is an essential step in rendering the problem numerically tractable. The algorithm works by heuristically exploiting any existing sparsity of the original problem, in particular its objective function, in order to automatically construct a sparse outer approximation. In our particular case, the sparse SDP resulting from our algorithm was not only far simpler, it was in fact an exact sparse representation of the original problem. Below, we also discuss the general conditions required for this to be the case for our algorithm. We highlight that it is potentially applicable to many other SDPs, providing at least an outer approximation, although there are no guarantees that a sparse simplification of the original problem will be obtained.

To see how a problem's sparsity can be exploited, consider a standard form SDP over $\Complexes^{n \times n}$:
\begin{align}\begin{split}
	\textbf{Given:} &\quad F, \; (C^{(k)})_k, \; (b_k)_k \\
	\textbf{Find:} &\quad \max_X 
	\; \langle F, X \rangle \\
	\textbf{Subject to:} &\quad \langle C^{(k)}, X \rangle = b_k, \quad k = 1,...,m \\
	&\quad X \ge 0.
\end{split}\end{align}
where $\langle A, B \rangle = \Tr{A^\dagger B}$. We assume $F$ is Hermitian so that $F_{ij} \ne 0$ implies $F_{ji} \ne 0$. Our goal is to exploit any sparsity of $F$, which defines the objective function, to simplify this problem by obtaining a suitable sparse version. This is achieved by finding a subset of variables $X_{ij}$ which is self-sufficient to solve the problem, and this is done by keeping track of the indices $(i,j)$ which can influence the objective function, directly or indirectly.

We begin by defining the initial \emph{base sparsity} of the problem as
\begin{equation}
	\theta_\text{base}^0 := \set{ (i,j) \cond F_{ij} \ne 0 } \union \{ (i,i) \}_{i=1}^{n},
\end{equation}
i.e., the pairs of indices with non-zero entries in $F$, together with the indices in the diagonal. We include the diagonal indices by default in order to ensure $X \ge 0$ can be satisfied properly in the sparse version of the problem, as will be explained later. As $\langle F, X \rangle = \sum_{i,j} F^{*}_{ij} X_{ij}$, the set $\theta_\text{base}^0$ corresponds to the entries of $X$ which appear explicitly in the objective function, i.e., the optimization variables $X_{ij}$ affecting its value, together with the diagonal entries. However, these are not the only variables relevant to the problem, as they also appear in linear constraints involving other variables not present in the objective, which may affect the objective function indirectly. Our next goal is to collect all of these indirect constraints to the objective function's variables. Defining the subset of constraints $C^{(k)}$ which relate to $\theta_\text{base}^0$ as
\begin{equation}
	\cK^0 := \bigcup_{(i,j) \in \theta_\text{base}^0} \set{k \cond C^{(k)}_{ij} \ne 0 },
\end{equation}
we can define the initial \emph{extended sparsity}
\begin{equation}
	\theta_\text{ext}^0 := \bigcup_{k \in \cK^0} \set{ (i,j) \cond C^{(k)}_{ij} \ne 0 }.
\end{equation}
In words, this set includes the indices for any variables present in a linear constraint also involving the variables directly affecting the objective function. The extended sparsity $\theta_\text{ext}^0$ specifies an initial guess for the minimum subset of variables $\{ X_{ij} \;\vert\; (i,j) \in \theta_\text{ext}^0 \}$ which should be considered in the SDP. The positivity constraint $X \ge 0$ could then be specified, in principle, in terms of a block-diagonal matrix which contains the variables specified by $\theta_\text{ext}^0$ within its blocks. Positivity of this block-diagonal matrix, and in turn of $X$, could then be achieved by enforcing positivity for each block separately.

As positivity constraints require significant computational resources, ideally these blocks should be made as small as possible. In order to obtain the minimal block-diagonal form, we first use $\theta_\text{ext}^0$ to construct an adjacency matrix $A^0$
\begin{equation}
	[A^0]_{ij} := \begin{cases}
		1, \quad \text{if}\; (i,j) \in \theta_\text{ext}^0 \\
		0, \quad \text{else}
	\end{cases},
\end{equation}
with which we construct an undirected graph $G(A^0)$ with $n$ vertices, where each edge corresponds to a pair of indices $(i,j)$. We can use this graph to obtain a permutation which leads to the minimal block-diagonal form, as each block corresponds to a disjoint subset of indices which are interrelated. This in turn, corresponds to the notion of \emph{connected components} of a graph, i.e., its disjoint subgraphs, as shown in \cref{fig:adjacencygraph}. Let $\{g_r\}_{r=1}^R$ be the $R$ connected components of $G(A^0)$, so that $G(A^0) = \bigcup_{r=1}^{R} g_r$, and let $\nu_r$ be the set of indices of the vertices in $g_r$. We may use the $\nu_r$ to construct a permutation $\pi = \left( \nu_1 \; \nu_2 \; \cdots \; \nu_R \right)$, together with its associated permutation matrix $P$, such that $P A^0 P^\T$ is block-diagonal, with one block per connected component. E.g., from \cref{fig:adjacencygraph}, we could use the permutation $\pi = [(5 2 7 3) (4) (1 6)]$, where the ordering of the $\nu_r$, or within each $\nu_r$, are irrelevant. This permutation leads to the block-diagonal adjacency matrix shown in figure \cref{fig:adjacencygraphblocks}, on the left.

\begin{figure}[h]\centering
	\includegraphics[width=0.4\linewidth]{"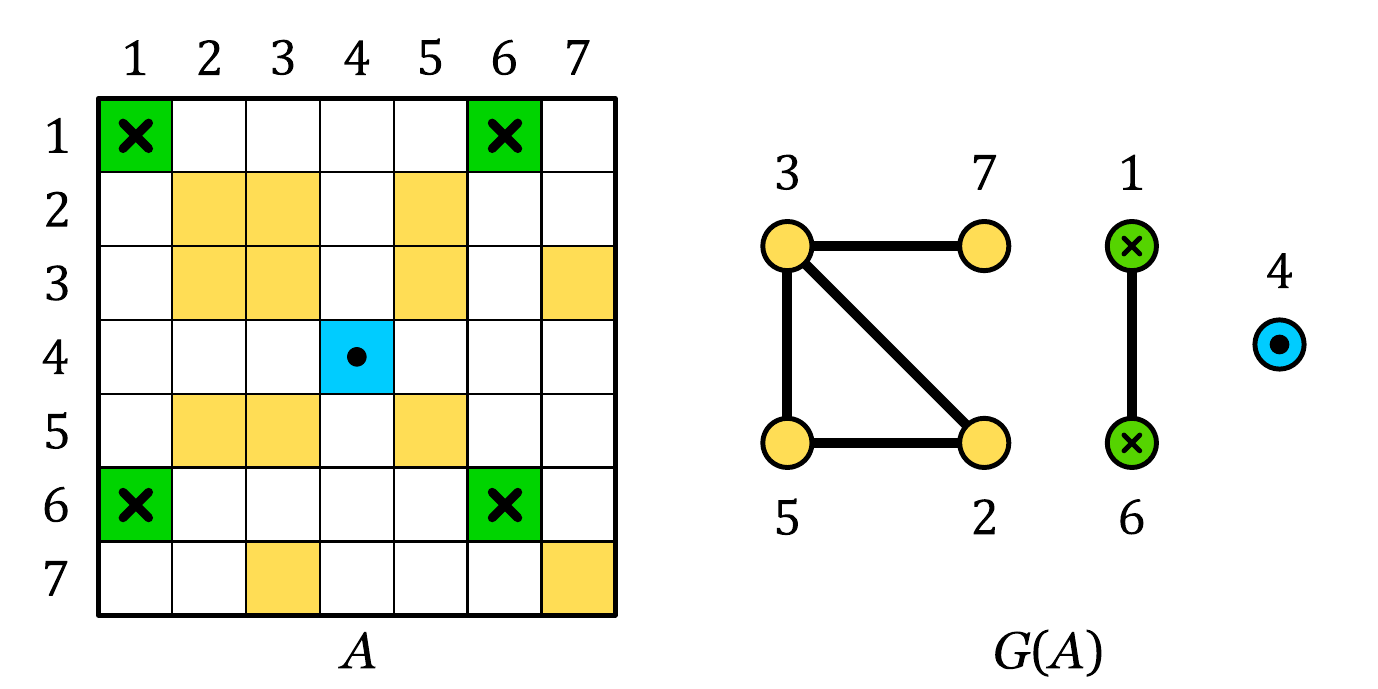"}
	\caption{An adjacency matrix $A$ and its corresponding graph $G(A)$, with three connected components. Non-zero entries are depicted with the colors (and markers) of the component they belong to. Self-edges, corresponding to $(i,i)$ entries, are omitted.}
	\label{fig:adjacencygraph}
\end{figure}

Therefore, if we take $\theta_\text{ext}^0$ to be the sparsity pattern for $X$, we can construct the sparse matrix
\begin{equation}
	\widetilde{X} := P^\T \left( \bigoplus_{r = 1}^{R} \widetilde{B}_r \right) P, \quad \widetilde{B}_r \ge 0 \; \forall \, r,
\end{equation}
with sparse blocks of variables given by $\widetilde{B}_r$. However, imposing positivity for each block is not always possible in this case, as the blocks themselves may have a sparse structure incompatible with positivity constraints. Concretely, there might be \emph{no} $\widetilde{X} \ge 0$ matrix with sparsity structure given by $\widetilde{X}_{ij} = 0 \; \forall\; (i,j) \notin \theta_\text{ext}^0$.

\begin{figure}[h]\centering
	\includegraphics[width=0.4\linewidth]{"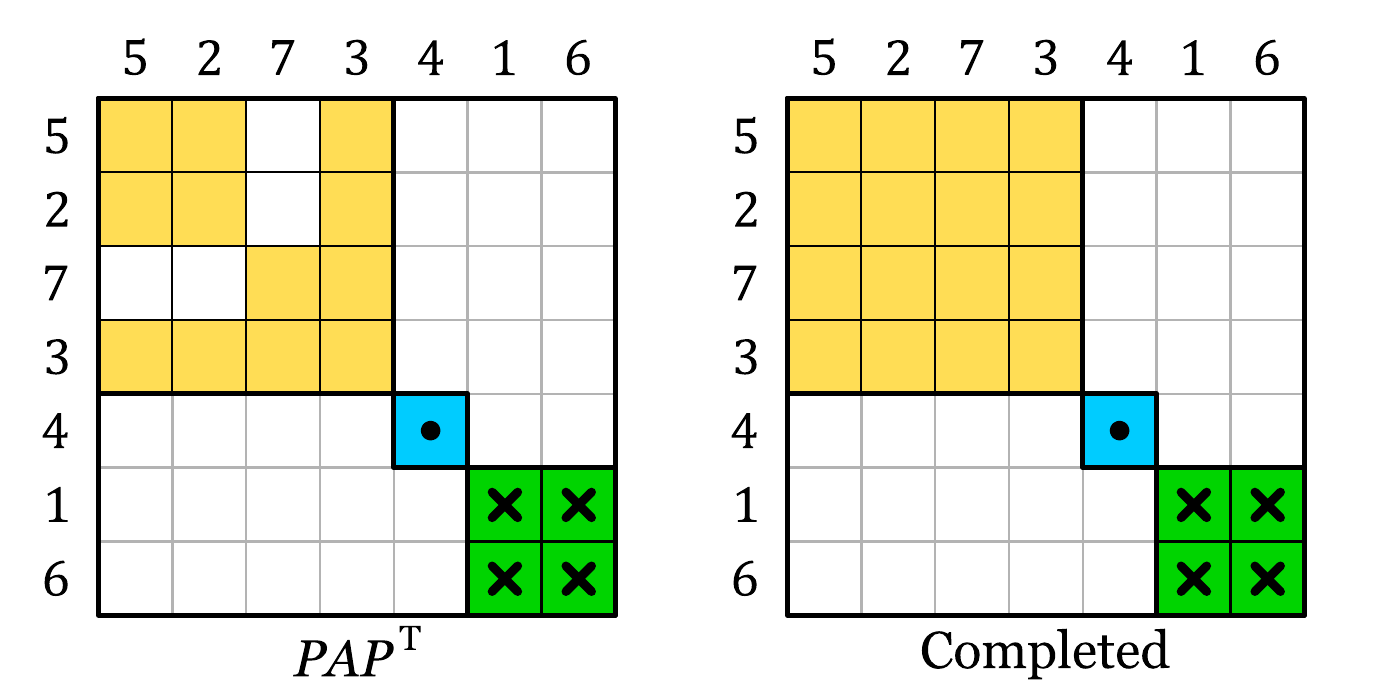"}
	\caption{(Left) Under the permutation $\pi = [(5 2 7 3) (4) (1 6)]$, the matrix $A$ of \cref{fig:adjacencygraph} becomes block-diagonal, with blocks possibly containing ``gaps'' indicating that not all pairs of vertices share an edge for that component. (Right) The blocks can be completed by adding the missing entries (i.e., edges).}
	\label{fig:adjacencygraphblocks}
\end{figure}

To address this in general, we instead assume all the blocks are dense. We can then define the initial \emph{completed sparsity}
\begin{equation}
	\theta_\text{comp}^0 := \bigcup^{R}_{r = 1} \set{ (i,j) \in \nu_r \times \nu_r },
\end{equation}
by including all missing indices $(i,j)$ within each block. In graph theoretical terms, we turn each connected component into a complete graph, i.e., every pair of vertices shares an edge, and use the non-zero entries of the resulting adjacency matrix to define $\theta_\text{comp}^0$; see \cref{fig:adjacencygraphblocks}, right panel. Denoting by $B_r$ the dense block matrices, we can define the first sparse approximation of $X$ as
\begin{equation}
	\widetilde{X}^0 := P^\T \left( \bigoplus_{r = 1}^{R} B_r \right) P, \quad B_r \ge 0 \; \forall \, r,
\end{equation}
such that $\widetilde{X}^0 \ge 0$ can be safely imposed through block-positivity, with no risk of rendering the SDP unfeasible.

The next step is to analyze how this sparse structure affects the other constraints in the SDP. These procedures expanded our initial base sparsity $\theta_\text{base}^0$ by including a larger set of indices (and therefore, variables) of the original dense $X$, resulting in $\theta_\text{comp}^0$. However, these additional variables may themselves appear in other linear constraints, involving even further variables we have not yet taken into account. These inter-dependencies between variables lead to a cascade of new constraints and variables that must be considered, i.e., the additional indices included in $\theta_\text{comp}^0$ specify variables which appear in constraints which were previously ignored, as they did not involve variables specified by $\theta_\text{base}^0$. Therefore, we repeat the above procedures until a stable and completely self-contained set of indices naturally emerges, i.e., we define a new base sparsity $\theta^1_\text{base} := \theta^0_\text{comp}$, and repeat the process $n+1$ times until $\theta^{n+1}_\text{base} = \theta^{n}_\text{base}$, at which point we call this stable set the \emph{effective sparsity} $\theta^{*}$ for the problem. In our particular applications, the algorithm converges within 2 to 4 iterations.

Note that, by construction, any variable $X_{ij}$ for $(i,j) \in \theta^{*}$ appears in the objective function directly, or if not, affects the variables in the objective either directly (through linear constraints) or indirectly (through positivity constraints and, in turn, through further constraints, and so on, recursively).
The goal of the above algorithm is to automatically discover a set of variables and constraints which is self-sufficient, given any objective function and linear constraints for an SDP.

While the algorithm generally provides an outer approximation, it can also lead to an exact sparse version of the original problem if the unused linear constraints are homogeneous, which is the case for the problem we consider. To understand why, we define the set of unused constraints as
\begin{equation}
	\cK^*_{\perp} := \set{k \cond C^{(k)}_{ij} = 0, \; \forall \; (i,j) \in {\theta^{*}} }.
\end{equation}
If the unused constraints are all homogeneous, i.e., $b_k = 0$ for all $k \in \cK^*_{\perp}$, then the original dense problem can be solved entirely within $\theta^{*}$, without affecting the objective function, as any sparse solution can be made into a feasible solution to the dense problem. This can be understood as follows. Let $X^{*}$ be a solution to the sparse problem, and $X^{*}_\perp$ a matrix in the complementary space $\theta^{*}_{\perp}$, containing all indices not appearing in $\theta^{*}$. Any constraints discarded by the algorithm (i.e., the indices in $\cK^*_{\perp}$) will act entirely on $\theta^{*}_{\perp}$. Then, as the objective function contains only variables within $\theta^{*}$, we may write a dense solution $X = X^{*} + X^{*}_{\perp}$, such that
\begin{equation}
	\Tr{F^\dagger X} = \Tr{F^\dagger X^{*}} + \Tr{F^\dagger X^{*}_{\perp}} = \Tr{F^\dagger X^{*}},
\end{equation}
with $X^{*}, X^{*}_{\perp} \ge 0$, and such that all constraints $\langle C^{(k)}, X \rangle = b_k$ uncouple onto orthogonal spaces. Now, if the linear constraints involving $X^{*}_{\perp}$ are all homogeneous ($b_k = 0$ for all $k \in \cK^*_{\perp}$), we may set $X^{*}_{\perp} = 0$ directly, while simultaneously satisfying all constraints in the dense problem, and without affecting the objective function. Therefore, any sparse solution $X^{*}$ is also a valid solution to the dense problem, with the same value for the objective: the sparse problem is exact. Note that the diagonal indices needed to be included in $\theta_\text{base}^0$ in order to ensure $X \ge 0$ is satisfied because of $X^{*} \ge 0$, as otherwise positivity of $X$ could always be satisfied for any $X^{*}$ by choosing arbitrarily large diagonal entries in $X^{*}_{\perp}$.

\medskip

Despite the success of this algorithm in giving a sparse solution to our particular problem, in general, there is no guarantee that $\theta^{*}$ will be sparse. It may be the case that the algorithm eventually includes all entries of $X$, in which case an alternative approach must be used, e.g., stopping the algorithm before a stable sparsity pattern is reached (generally resulting in an outer approximation), using a more refined completion procedure (e.g., block-wise chordal completions~\cite{zheng2021,fukuda2001}), or exploiting additional structures of the original problem. Moreover, the sparsity of the original problem, and thus the efficacy of the above algorithm, relies heavily on the choice of basis in which the problem is represented. Fortunately, for our choice of parameters, the objective function and the resulting $\theta^{*}$ were sufficiently sparse for the SDP to be numerically tractable. Furthermore, as the partial trace constraints $\TrP{\SO_1}{{V_N}^\dagger \phi V_N} = \frac{\one_{\SI_1}}{\dES} \otimes \TrP{\SA{1}}{{V_N}^\dagger \phi V_N}$ are all homogeneous, this approach lead to an exact sparse representation of \cref{eq:finalSDPmain}.

Finally, we note that other techniques exist to formulate sparse SDPs~\cite{zheng2021,fukuda2001}, typically relying on the problem having an inherent ``aggregate sparsity'' (the joint sparse structure of $F$ and all $C^{(k)}$ taken together), or being analytically amenable to a sparser representation by an appropriate change of variables. In our problem, no such structure is a priori apparent: only the objective is inherently sparse, with the linear constraints jointly involving all variables in a non-trivial manner. In such cases, our algorithm is then a suitable technique to obtain an ``effective'' sparsity focused on the objective function's inherent sparse structure, with the advantage of being more agnostic to the problem's overall structure.


\section{Maximum probability for closed systems}\label{app:trivialbound}

Here, we provide a proof for the analytical bounds for the $\dE=1$ case mentioned in \cref{sec:opensys}. To this end, recall that, in the sequential measurement protocol, the probability of a sequence of outcomes can be written as
\begin{equation}
\label{eqn::AppProbFac}
	p(\va|\dE) = p(a_1, \dots, a_L|\dE) = \Tr{\cM_{a_L} \!\! \circ \cU \! \circ \dots \circ \cM_{a_1} \!\! \circ \cU(\rhoEz \! \otimes \! \rhoSz)},
\end{equation}
with the measure-and-prepare map $\cM_{a}(\rho_{ES}) = \TrP{S}{\rho_{ES} \cdot (\one_E \otimes E^a_S)} \otimes \rhoSz$. If $\dE = 1$, the unitary corresponds only to a local rotation on the system, and can be embedded in the state $\rhoSz$ and measurements $E_S^a$. Thus, the maximum depends only on the state and measurements. To compute it, we first note that since both $\rhoSz$ and $(E_S^a)_a$ are the same at each step and $\cM_a$ is a measure and prepare operation, all outcomes are independent and identically distributed. Writing $q_a = \Tr{\rhoSz E^a_S}$ as the probability of outcome $a$, with $\sum_a q_a = 1$, and using $n_a$ as the number of occurrences of a symbol $a$ in $\va$, we can write the probability as:
\begin{equation}
	p(\va|\dE=1) = \prod_{\ell=1}^{L} q_{a_\ell} = \prod_{a \in \cA} q^{n_a}_a, \qquad \text{with} \qquad n_a \in \Naturals, \quad \sum_{a \in \cA} n_a = L, \qquad  q_a > 0, \;\text{and}\; \sum_{a \in \cA} q_a = 1.
\end{equation}
For simplicity, we assume that $\va$ contains every symbol of $\cA$ at least once, such that $n_a, q_a > 0$. Otherwise, we could simply assume a smaller $\cA$ where this is the case. The maximum probability $\omega_1^{\va}$ can be found with standard techniques, such as Lagrange multipliers. Using the Lagrangian
\begin{equation}
	\cL = \prod_{a \in \cA} q^{n_a}_a - \alpha \left( \sum_{a \in \cA} q_a - 1 \right),
\end{equation}
we calculate the partial derivatives for each $q_a$, and equate them to zero:
\begin{equation}
	\frac{\partial \cL}{\partial q_a} = \left( \frac{n_a}{q_a} \right) p - \alpha = 0, \; \forall a.
\end{equation}
Since $q_a > 0$ for all $a$, $p > 0$, we may rewrite this as:
\begin{equation}
	\frac{n_a}{q_a} = \frac{\alpha}{p} = \gamma \quad \to \quad q_a = \frac{n_a}{\gamma}, \; \forall a.
\end{equation}
Summing over $a$ and applying the constraints $\sum_a q_a = 1$ and $\sum_a n_a = L$, we obtain:
\begin{equation}
	\sum_a q_a = \frac{1}{\gamma} \sum_a n_a  \quad \to \quad \gamma = L.
\end{equation}
Thus, the maximum value is
\begin{equation}
	\omega_1^{\va} = \prod_{a\in\cA} \left( \frac{n_a}{L} \right)^{n_a}, \qquad \text{for} \quad q_a = \frac{n_a}{L}, \; \forall a
\end{equation}
Note that this solution is unique if $q_a > 0$ and $p > 0$. The only alternative solutions would require that $q_a = 0$ for some $a$, and thus $p = 0 < \omega_1^{\va}$. Therefore, this is indeed the global maximum.

We emphasize that this direct calculation of $\omega_1^{\va}$ crucially depends on the specific form of the instrument we employ. If the respective outcomes do \emph{not} correspond to a $\cM_{a}$ of measure-and-prepare form, the joint probabilities $p(\va|\dE)$ do not factorize like in \cref{eqn::AppProbFac}, and an alternative approach must be used.

\end{widetext}

\bibliography{biblio}
\end{document}